%% file: main20180514-arxiv.tex
\documentclass[12pt]{article}
\usepackage[top=1.25in,bottom=1in,left=1in,right=1in]{geometry}
\usepackage{graphicx}
\usepackage{amsmath}
\usepackage{bbm}
\usepackage{amssymb}
\usepackage{amsthm}
\usepackage{wrapfig}
\usepackage{multicol}
\usepackage{float}
\usepackage{pdfpages}
\usepackage[colorlinks=true, urlcolor=blue, linkcolor=blue,citecolor=blue]{hyperref}
\usepackage{natbib}
\usepackage[utf8]{inputenc}
\usepackage{subfiles}
\usepackage{blindtext}
\usepackage{titlesec}
\usepackage{lipsum}
\usepackage[title,toc,titletoc]{appendix}
\usepackage{etoolbox}
\usepackage{tikz}
\usepackage[font=footnotesize, textfont=it]{caption}

\usepackage[]{color-edits}
\addauthor{vs}{red}
\addauthor{al}{blue}
\addauthor{xs}{brown}

\newcommand{\RN}[1]{  \textup{\uppercase\expandafter{\romannumeral#1}}}

\newtheorem{definition}{Definition}
\newtheorem{proposition}{Proposition}
\newtheorem{theorem}{Theorem}
\newtheorem{corollary}{Corollary}
\newtheorem{lemma}{Lemma}
\newtheorem{fact}{Fact}
\newtheorem{assumption}{Assumption}

\theoremstyle{remark}
\newtheorem{example}{Example}
\theoremstyle{remark}

\DeclareMathOperator*{\argmin}{argmin}

\DeclareMathOperator*{\diag}{diag}

\title{Optimal and Myopic Information Acquisition\thanks{This paper was previously circulated under the title of ``Dynamic Information Acquisition from Multiple Sources." We thank Yash Deshpande,  Johannes H\"orner, Michihiro Kandori, Elliot Lipnowski,  George Mailath, Juuso Toikka, Anton Tsoy, and Muhamet Yildiz for helpful comments. We are especially grateful to Drew Fudenberg, Eric Maskin and Tomasz Strzalecki for their constant guidance and support.} \\}

\author{Annie Liang\thanks{Department of Economics, University of Pennsylvania} \quad \quad Xiaosheng Mu\thanks{Department of Economics, Harvard University} \quad \quad Vasilis Syrgkanis\thanks{Microsoft Research}}

\begin{document}

\maketitle
\begin{abstract}
We consider the problem of optimal dynamic information acquisition from many correlated information sources. Each period, the decision-maker jointly takes an action and allocates a fixed number of observations across the available sources. His payoff depends on the actions taken and on an unknown state.  In the  canonical setting of jointly normal information sources, we show that the optimal dynamic information acquisition rule proceeds myopically after finitely many periods. If signals are acquired in large blocks each period, then the optimal rule turns out to be myopic \emph{from period 1}. These results demonstrate the possibility of robust and ``simple" optimal information acquisition, and simplify the analysis of dynamic information acquisition in a widely used informational environment.\\

\noindent \textbf{JEL\ codes:}\ C44, D81, D83


\end{abstract}
\clearpage

\pagebreak

\tableofcontents

\pagebreak

\section{Introduction}

In a classic problem of sequential information acquisition, a Bayesian decision-maker (DM) repeatedly acquires information and takes actions. His payoff depends on the sequence of actions taken, as well as on an unknown payoff-relevant state. We consider a setting in which the DM acquires information from a limited number of \emph{flexibly correlated} information sources, and allocates a fixed number of observations across these sources each period. The optimal strategy for information acquisition is of interest.

Neglecting dynamic considerations, a simple strategy is to acquire at each period the signal (or the set of signals) that maximally reduces uncertainty about the payoff-relevant state. We refer to this as the \emph{myopic} (or greedy) rule, as it is the optimal rule if the DM mistakenly believes each period to be the last possible period of information acquisition.

This myopic rule turns out to possess strong optimality properties in a widely used setting. Suppose that the available signals are \emph{jointly normal}. If signal observations are acquired in sufficiently large blocks each period, then myopic information acquisition is optimal from period 1 (Theorem \ref{propBatch}). We provide a sufficient condition on the required size of the block; this condition depends on primitives of the informational environment but not on the payoff function. Theorem \ref{propSeparable} characterizes a condition on the prior and signal structure, given which myopic information acquisition is optimal from period 1 for \emph{any} block size. In both of these cases, the optimal information acquisition strategy can be exactly and simply characterized. Additionally, instead of solving for the optimal decision rule and information acquisition strategy jointly (as would otherwise be required), our results show that one can separate these two problems, and  solve for the optimal decision rule in this setting as if information acquisition were exogenous.  

Finally, for generic signal structures, and for any block size, the optimal strategy proceeds by myopic acquisition after finitely many periods (Theorem \ref{propGenericEventual}). These results hold across all payoff functions (and in particular, independently of discounting); thus, myopic information acquisition is (eventually) ``robustly" best.

Why does the myopic rule perform so well? The main inefficiency of myopic planning is that it neglects potential complementarities across signals. A signal that is individually uninformative can be very informative when paired with other signals; thus, repeated (greedy) acquisition of the best single signal need not result in the best sequence of signals.\footnote{See Section \ref{sec:intuition} for a concrete example and further discussion.} A key observation is that whether the DM perceives two signals as providing complementary information depends on his current belief over the state space.\footnote{As a simple example, suppose the payoff-relevant state is $\theta_1$, and the available signals are about $\theta_1+\theta_2$ and $\theta_2$. These signals are ``complementary" when the agent's prior belief is that $\theta_1$ and $\theta_2$ are independent: observations of the first signal improve the value of observing the second signal, and vice versa. But suppose the DM's prior is such that $\theta_2=\theta_1$; then, these signals in fact become substitutes.} This means that complementarities across signals are not intrinsic to the underlying signal correlation structure: As the DM's beliefs about the states evolve, so too do his perceptions of the correlations across signals. It is clear that as information accumulates, the DM's beliefs become more precise about each of the unknown states. This does not itself lead to optimality of myopic information acquisition (see Section \ref{sec:precisionCorrelation} for more detail). We show that the key force comes from a second effect of information accumulation: The DM's beliefs evolve in such a way that the signals \emph{endogenously de-correlate} from  his perspective, and are eventually perceived as providing approximately independent information. At the limit in which all signals are independent, the value of any given signal can be evaluated separately of the others. The dynamic problem is  thus ``separable," and can be replaced with a sequence of static problems. Given sufficiently many signal observations, we have only approximate separability, which we show is sufficient for the myopic rule to be optimal. 

The mechanism we identify is different from the one underlying a classic result from the experimentation literature.  In ``learning by experimentation" settings, myopic behavior is eventually near-optimal: in the long run, the DM's beliefs converge, so the value of exploration (i.e.\ learning) becomes second-order relative to the value of exploitation of the perceived best arm.\footnote{\cite{EasleyKiefer} and \cite{Aghion} show that  if there is a unique myopically optimal policy at the limiting beliefs, then the optimal policies must converge to this policy. In our setting, every policy (signal choice) is trivially myopic at the limiting beliefs (a point mass at the true parameter), so we do not have uniqueness and cannot use this argument to identify long-run behavior.} In our paper, signal acquisition decisions are driven by learning concerns exclusively, as there is by design no exploitation incentive. To see this, recall that in the classic multi-armed bandit problem \citep{Gittins1979,Bergemann2008}, actions play the dual role of influencing the evolution of beliefs and also determining flow payoffs. In our setting (which does not fall into the multi-armed bandit framework), there is a separation between \emph{signal choices}, which influence the evolution of beliefs, and \emph{actions}, which determine (unobserved) payoffs. Myopic signal choices become optimal in our framework because \emph{they maximize the speed of learning}, and not because they optimize a tradeoff between learning and payoff. (Additionally, a myopic strategy is immediately optimal in multi-armed bandit problems only under very restrictive assumptions \citep{BerryFristedt, BanksSundaram}.)

Our results simplify the analysis of optimal dynamic information acquisition in an informational environment that is commonly used in economics: normal signals. However, the core of our analysis\textemdash the ``endogenous de-correlation" of signals described above\textemdash does not rely on the assumption of normality. As we discuss further in Section \ref{sec:intuition}, this de-correlation derives from a Bayesian version of the Central Limit Theorem, which holds for arbitrary signal distributions. This suggests that  our eventual optimality result (Theorem \ref{propGenericEventual}) generalizes.\footnote{Specifically, we conjecture that for general signals, the optimal rule eventually proceeds myopically when we restrict to certain decision problems (e.g.\ prediction of the payoff-relevant state). Immediate optimality of the myopic rule given sufficiently many signals, and also the independence of our results to the payoff function, do rely on properties of the normal environment (see Section \ref{sec:normality} for further details).}

We conclude by demonstrating several extensions. To facilitate application of our results, we extend our environment to a multi-player setting in which individuals privately acquire information before playing a one-shot game at a random final period. This extension connects our results to a literature on games with Gaussian information \citep{HellwigVeldkamp,MyattWallace,Pavan,LambertOstrovskyPanov}.\footnote{For games of information acquisition beyond the Gaussian setting, see e.g.\ \cite{Persico}, \cite{BergemannValimaki}, \cite{Yang} and \cite{Denti}. All of these papers restrict to a single signal choice.} We present corollaries of our main results, and use these to extend results from \cite{HellwigVeldkamp} and \cite{LambertOstrovskyPanov} in an online appendix.
 Finally, we demonstrate extensions to environments with choice of information ``intensity" (the number of signals to acquire each period), to multiple payoff-relevant states (for a class of prediction problems) and to a continuous-time setting.  

Our work primarily builds on a large literature about optimal dynamic information acquisition \citep{MoscariniSmith,FudenbergStrackStrzalecki,CheMierendorff,Mayskaya,SteinerStewartMatejka,HebertWoodford,Zhong} and a related literature on sequential search \citep{Wald,ArrowBlackwellGirshick,Weitzman,Callander,KeVillas-Boas,Bardhi}. In contrast to an earlier focus on optimal stopping and choice of signal precision, our framework characterizes choice between different \emph{kinds} of information, as in the work of \cite{FudenbergStrackStrzalecki} (where the sources are two Brownian motions), and  \cite{CheMierendorff} and \cite{Mayskaya} (where the sources are two Poisson signals).\footnote{\cite{CheMierendorff} and \cite{Mayskaya} consider choice between two Poisson signals, each of which provides evidence towards/against a binary-valued state. The Poisson model (with a finite state space) is more suited to applications such as learning about whether a defendant is guilty or innocent, while the Gaussian model describes for example learning about the (real-valued) return to an investment.} Compared to this work, we allow for many (i.e.\ more than two) sources with flexible correlation.\footnote{\cite{Callander} also emphasizes correlation across different available signals. But the signals in \cite{Callander} are related by a Brownian motion path, which yields a special correlation structure. Similar models are studied in \cite{GarfagniniStrulovici} and in \cite{Bardhi}.}

Another strand of the literature considers a DM who chooses from completely flexible information structures at entropic (or more generally, ``posterior-separable") costs, such as in \cite{SteinerStewartMatejka}, \cite{HebertWoodford} and \cite{Zhong}. Compared to these papers, our DM has access to a \emph{prescribed} and limited set of signals.\footnote{In Section \ref{ex:intensity}, we do allow the DM to also control the intensity of information acquisition by endogenously choosing how many signals to acquire in each period. But even in that extension, we assume that the incurred information cost is a function of the number of observations. This is analogous to \cite{MoscariniSmith} and is distinguished from the above papers that measure information cost based on belief changes.} 

Finally, acquisition of Gaussian signals whose means are linear combinations of unknown states appears previously in the work of \cite{MeyerZwiebel} and \cite{SethiYildiz}. In particular, \cite{SethiYildiz} characterizes the long-run behavior of a DM who myopically acquires information from experts with independent biases. See Section \ref{relatedLit} for remaining connections to the literature. See Section \ref{relatedLit} for remaining connections to the literature.

\section{Preliminaries} \label{prelim}

\subsection{Model}
Time is discrete. At each time $t=1, 2, \dots$, the DM first chooses from among $K$ information sources, and then chooses an action $a_t$ from a set $A_t$.\footnote{Thus, the action $a_t$ can be based on the information received in period $t$. The timing of these choices is not important for our results.}

The DM's payoff $U(a_1, a_2, \dots ; \theta_1)$ is an arbitrary function that depends on the sequence of action choices and a payoff-relevant state $\theta_1 \in \mathbb{R}$. We assume that payoffs are realized only at an (exogenously or endogenously determined) end date; thus, the information sources described below are the only channel through which the DM learns. This assumption distinguishes our model from multi-armed bandit problems, see Section \ref{relatedLit} for further discussion.

Stylized cases of such decision problems include: \\
\textbf{Exogenous Final Date.} An action is taken just once at a final period $T$ that is determined by an arbitrary distribution over periods.\footnote{Special cases include \emph{geometric discounting}, in which every period (conditional on being reached) has a constant probability of being final, as well as \emph{Poisson arrival} of the final period.
}  The DM's payoff is $U(a_1,a_2,\dots;\theta_1)=u_T(a_T, \theta_1)$ where $T$ is the (random) final time period, and $a_T$ is the action chosen in that period. The time-dependent payoff function $u_T(a_T,\theta_1)$ may incorporate discounting.

\smallskip

\noindent \textbf{Endogenous Stopping with Per-Period Costs.} Take each action $a_t$ to specify both the decision of whether to stop, and also the action to be taken if stopped. The DM's payoff is $U(a_1,a_2,\dots;\theta_1)= u_T(a_T,\theta_1)$ where $T$ is the (endogenously chosen) final time period, and $a_T$ is the action chosen in that period. The payoff-function $u_T(a_T,\theta_1)$ may incorporate discounting and/or a per-period cost to signal acquisition. Costs are fixed across sources in a given period, but can vary across periods.\footnote{See e.g. \cite{FudenbergStrackStrzalecki} and \cite{CheMierendorff} for recent models with constant waiting cost per period.}

\smallskip

Apart from the decision problem, there are $K$ information sources, which depend on the unknown and persistent state vector $\theta = (\theta_1, \dots, \theta_K)' \sim \mathcal{N}(\mu^0, V^0)$.\footnote{Here and later, we exclusively use the apostrophe to denote vector or matrix transpose.} This vector includes the payoff-relevant unknown $\theta_1$ and additionally $K-1$ \emph{payoff-irrelevant} unknowns $\theta_2, \dots, \theta_K$. The role of these auxiliary states is to permit correlations across the information sources conditional on $\theta_1$; this allows, for example, for the sources to be afflicted by common and persistent biases. In each period, the DM chooses $B$ sources (allowing for repetition), where $B \in \mathbb{N}^{+}$ is interpreted as a fixed time/attention constraint (see Section \ref{ex:intensity} for extension to endogenous choice of $B$). Choice of source $k=1,\dots, K$ produces an observation of
\[
X_k = \langle c_k, \theta \rangle + \epsilon_k, \quad \epsilon_k \sim \mathcal{N}(0,\sigma_k^2)
\]
where the coefficient vectors $c_k = (c_{k1}, \dots, c_{kK})'$ and signal variances $\sigma_k^2$ are fixed (and known), but the Gaussian error terms are independent across realizations and sources. Throughout, we use $C$ to denote the matrix of coefficients whose $k$-th row is $c_k'$. 

We impose the following assumption on the informational environment: 
\begin{assumption}[Non-Redundancy]\label{assumptionIdentifiability} The matrix $C$ has full rank, and no proper subset of row vectors of $C$ spans the coordinate vector $e_1'$. Equivalently, the inverse matrix $C^{-1}$ exists, and its first row consists of non-zero entries.
\end{assumption}

\noindent Heuristically, this means that the DM \emph{can and must} observe each source infinitely often to recover the value of the payoff-relevant state $\theta_1$. Since in this paper we restrict the number of sources to be the same as the number of unknown states, the above assumption is generically satisfied.\footnote{Throughout, ``generic" means with probability $1$ for randomly drawn coefficient matrices $C$.}$^,$\footnote{Although we have assumed that the number of sources and signals are the same, our results extend to cases in which there are fewer sources than states, so long as 
$e_1'$ is spanned by the whole set of signal coefficient vectors and not by any proper subset.}

Assumption of \emph{no redundant sources} simplifies our analysis, as it guarantees that all sources will be sampled infinitely often under the criteria we consider. With redundant sources, a new question emerges regarding which subset of sources the DM will choose from. Characterization of that subset is the focus of \cite{LiangMu2017}.

\subsection{Interpretations}
We provide below several interpretations of this framework. 

\emph{News Sources with Correlated Biases.} On election day $T$, a DM will choose which of two candidates $I$ and $J$ to vote for, where his payoff depends on $\theta_1=v_I-v_J$, the difference between the candidates' qualities $v_I$ and $v_J$. In each period up to time $T$, the DM can acquire information from different news sources. All sources provide biased information, and moreover the biases are correlated across the sources. As the DM acquires information, he learns not only about the payoff-relevant state $\theta_1$, but also how to de-bias (and aggregate) information from the various news sources.

\smallskip

\emph{Attribute Discovery.}  A product has $K$ unknown attributes $\tilde{\theta}_1, \dots, \tilde{\theta}_K$. Its value $\theta_1$ is some arbitrary linear combination of these attribute values. For example, the DM may want to learn the value of a conglomerate composed of several companies, where each company $i$ is valued at $\tilde{\theta}_i$ and the value of the conglomerate is $\theta_1:= \alpha_1 \tilde{\theta}_1 + \dots + \alpha_K \tilde{\theta}_K$. The DM has access to (noisy) observations of different linear combinations of the attributes; for example, he might have access to evaluations of each $\tilde{\theta}_i$ individually.\footnote{This model can be rewritten in our framework above, where the state vector is $(\theta_1, \tilde{\theta}_2, \dots, \tilde{\theta}_{K})$.} At some endogenously chosen end time, the DM decides whether or not to invest in the conglomerate.

\smallskip

\emph{Sequential Polling.}
A polling organization seeks to predict the average opinion in the population towards an issue. There are $K$ demographic groups in the population, and opinions in demographic group $k$ are normally distributed with unknown mean $\mu_k$ and known variance $\sigma_k^2$. The fraction of the population in each demographic group $k$ is $p_k$, so the average opinion is $\theta_1 := \sum_{k} p_k \mu_k$. It is not feasible to directly sample individuals according to the true distribution $p_k$, but the organization can sample individuals according to other non-representative distributions $\hat{p}_k \neq p_k$. Each period, the polling organization allocates a fixed budget of opinion samples across the available distributions (polling technologies), and posts a prediction for $\theta_1$. Its payoff is the average prediction error across some fixed number of periods. 

\smallskip

\emph{Intertemporal Investment.} Each action $a_t$ is a decision of how to allocate capital between consumption, and two investment possibilities: a  liquid asset (bond), and an illiquid asset (pension fund). The return to the liquid asset is known: $1$ dollar saved today is worth $e^{r}$ dollars tomorrow. The return to the illiquid asset is unknown, and it is the payoff-relevant state in the worker's problem; that is, every dollar invested today in the pension fund deterministically yields $e^{\theta_1}$ dollar(s) tomorrow. The worker works for $T$ periods, and in each of these periods he learns about $\theta_1$ (from some information sources) and then allocates his wealth across consumption, saving and investment. In period $T+1$, the worker retires and receives all the returns from his investments into the illiquid asset. His objective is to maximize the aggregated sum of his discounted consumption utilities and the payoff after retirement.\footnote{An important assumption of this example is that the return to investment is deterministic and only observed at the end. However, our model and results extend to a situation where there are ``free" signals arriving each period that do not count toward the capacity constraint $B$. By considering the realized log return as a particular free signal, the extension of our model covers the case where investment returns are stochastic and the DM observes past return realizations.}

\section{(Eventual) Optimality of Myopic Rule} \label{results}

\subsection{Myopic Information Acquisition} \label{sec:defineMyopic}
A strategy consists of an information acquisition strategy and a decision strategy. An \emph{information acquisition strategy}  is a measurable map from possible histories of signal realizations to multi-sets of $B$ signals, and a \emph{decision strategy} is a map from histories to actions. 

We will say that an information acquisition strategy is myopic if it proceeds by choosing signals that maximally reduce (next period) uncertainty about the payoff-relevant state.

\begin{definition}
An information acquisition strategy is \emph{myopic}, if at every next period, it prescribes choosing the $B$ signals that (combined with the history of observations) lead to the lowest posterior variance about $\theta_1$.
\end{definition}
\noindent Note that the $B$ signals which minimize the posterior variance also Blackwell dominate any other multi-set of $B$ signals (see e.g.\ \cite{HansenTorgersen}). Thus, myopic acquisition is optimal if the current period is the last chance for information acquisition, and this is true no matter what the payoff function is.

Our results below reveal a close relationship between the optimal information acquisition strategy and the myopic information acquisition strategy. We do not pursue a characterization of the optimal decision strategy, which in general depends on the payoff function, although we point out one application of our main results towards simplification of this characterization. 

\subsection{Main Results}
We present three results regarding optimality of the myopic information acquisition rule: Theorem \ref{propBatch} says that myopic information acquisition is optimal from period 1 if $B$ (the number of signals acquired each period) is sufficiently large. Our next two results hold for arbitrary $B$: Theorem \ref{propSeparable} provides a sufficient condition on the prior and the coefficient matrix $C$ under which myopic information acquisition is optimal from period 1, and Theorem \ref{propGenericEventual} states that the optimal rule is \emph{eventually} myopic in generic environments.

\begin{theorem}[Immediate Optimality under Many Observations]\label{propBatch}
Fix any prior and signal structure, and suppose $B$ is sufficiently large. Then the DM has an optimal strategy that acquires information myopically.\footnote{Without further assumptions on the payoff function $U$, we cannot assert \emph{strict} optimality of the myopic information acquisition strategy. For instance, this would not be true if there exists a ``dominant" action sequence that maximizes $U(a_1, a_2, \dots; \theta_1)$ for every value of $\theta_1$. But in most other cases, strictly more precise beliefs do lead to strictly higher expected payoffs, which implies unique optimality of the myopic rule.}
\end{theorem}

Optimally, the DM chooses the most informative $B$ signals in the first period based on his prior, then chooses the most informative $B$ signals in the second period based on his updated posterior, and so on. Note that since posterior variances are independent of signal realizations, and we have assumed that there is no feedback from actions, the above myopic strategy is \emph{history-independent}, and can be represented as a deterministic signal path. This implies that instead of solving for the optimal decision strategy and information acquisition strategy jointly (as would otherwise be required), one can solve for the optimal decision strategy with respect to an exogenous stream of information.\footnote{An application of this two-step approach (in continuous time) appears in the concurrent work of \cite{FudenbergStrackStrzalecki}, Section 3.5. See Section \ref{ex:contTime} for a brief discussion of how our model extends to continuous time.}

We additionally mention that Theorem \ref{propBatch} can be strengthened to optimality of  myopic information acquisition at all histories, including those in which the DM has previously deviated from the myopic rule. Finally, a precise bound for how large $B$ must be appears in Section \ref{sec:precisionCorrelation}.

Our next two results hold for arbitrary block sizes $B$. First, the myopic rule is again optimal from period 1 in a class of ``separable" environments. Let $f(q_1, \dots, q_K)$ denote the DM's posterior variance about $\theta_1$, updating from $q_i$ observations of each signal $X_i$. An informational environment is separable if its posterior variance function can be decomposed in the following way:

\begin{definition}
The informational environment $(V^0, C, \{\sigma_{i}^2\})$ is \emph{separable} if there exist convex functions $g_1, \dots, g_K$ and a strictly increasing function $F$ such that 
\[
f(q_1, \dots, q_K) = F(g_1(q_1) + \dots + g_K(q_K)).
\]
\end{definition}

\noindent Intuitively, separability ensures that observing signal $i$ does not change the relative value of other signals, but strictly decreases the marginal value of signal $i$ relative to every other signal.

Note that separability is not defined directly on the primitives of the informational environment ($V^0$, $C$, and $\{\sigma_i^2\}$), as it is based instead on the posterior variance function $f$. Nevertheless, $f$ can be directly computed from these primitives, and so is not an endogenous object.

The result below says that myopic information acquisition is optimal in all separable informational environments.

\begin{theorem}[Immediate Optimality in Separable Informational Environments]\label{propSeparable} 
Suppose the informational environment is separable. Then for every $B \in \mathbb{N}^{+}$, the DM has an optimal strategy that acquires information myopically. 
\end{theorem}

Separability encompasses several classes of informational environments that are of independent interest. For example:

\smallskip

\noindent \emph{Example} (Orthogonal Signals). The DM's prior is standard Gaussian ($V^0 = \mathbf{I}_K$), and the row vectors of $C$ are orthogonal to one another.\footnote{This is because the signals are independent from each other, see also Example 2 in Figure \ref{fig:intuition}.}

\smallskip

\noindent \emph{Example} (Multiple Biases). There is a single payoff-relevant state $\theta_1 \sim \mathcal{N}(0, v_1)$. The DM has access to observations of $X_1 = \theta_1 + \theta_2 + \dots + \theta_K + \epsilon_1$, where each $\theta_i$ ($i > 1$) is a persistent ``bias" independently drawn from $\mathcal{N}(0, v_i)$, and $\epsilon_1 \sim \mathcal{N}(0, \sigma_1^2)$ is a noise term i.i.d.\ over time. Additionally, the DM has access to signals about each bias 
$X_i = \theta_i + \epsilon_i$ ($i>1$), 
where $\epsilon_i \sim \mathcal{N}(0, \sigma_i^2)$.\footnote{The DM's posterior variance about $\theta_1$ is given by
\begin{align*}
f(q_1, \dots, q_K) & = v_0 -\frac{v_0^2}{v_0 +  \frac{\sigma_1^2}{q_1} + \sum_{i=2}^{K} \left(v_i - \frac{v_i^2}{v_i+ \sigma_i^2/q_i}\right)}.
\end{align*}
This can be written in the separable form.
}

\smallskip

In all remaining cases, optimal signal choices are eventually myopic. 

\begin{theorem}[Generic Eventual Myopia]\label{propGenericEventual}
Fix any prior covariance $V^0$ and signal variances $\{\sigma_i^2\}_{i=1}^{K}$. For generic coefficient matrices $C$, there exists a time $T^* \in \mathbb{N}$ that depends only on the informational environment. For this $T^*$, and for any $B$ and any decision problem, the DM has an optimal strategy that acquires information myopically after $T^*$ periods. 
\end{theorem}

\noindent That is, at all late periods, the optimal signal acquisitions are those that maximally reduce posterior variance in the given period.

The result above tells us that the optimal rule eventually proceeds by myopic signal acquisition; this is different from the statement that acquisition of signals myopically (from period 1) leads to the optimal signal path.  We show in Appendix \ref{appxMyopic} that this latter statement is also true. This complementary result (that the myopic signal path is eventually optimal) depends critically on Assumption \ref{assumptionIdentifiability} (no redundant signals). In particular, in environments with redundant signals, it is possible for the myopic and optimal signal paths to eventually sample from disjoint subsets of signals.\footnote{For example, suppose the available signals are $X_1=0.5 \theta_1 + \epsilon_1$, $X_2=\theta_1 + \theta_2 +\epsilon_2$, $X_3=\theta_1 - \theta_2 + \epsilon_3$, where noise terms are standard normal, states are independent, and prior variance about $\theta_2$ is larger than $3$. Myopic information acquisition (from period $1$) leads to exclusive sampling of signal $X_1$, while a patient DM eventually samples only from $X_2$ and $X_3$. This example is generalized in \citet{LiangMu2017}.} In contrast, we conjecture that Theorem \ref{propGenericEventual} (the optimal rule eventually proceeds myopically) extends beyond Assumption \ref{assumptionIdentifiability}. We leave this conjecture as an open question for future work. 

\section{Discussion} \label{discussion}

\subsection{Intuition for Theorems \ref{propBatch}-\ref{propGenericEventual}} \label{sec:intuition}
We begin with a simplified argument for two periods and $B=1$: Suppose that the best signal to acquire in period 1 is a part of the best pair of signals to acquire. In this case, no tradeoffs are necessary across the two periods, and it is optimal in both periods to acquire information myopically. In general however, signals that are individually uninformative (from the perspective of period 1) can be very informative as a pair; thus, myopic information acquisition in the first round can preclude acquisition of the best pair of signals.

At a high level, myopic information acquisition fails to be optimal when there are strong complementarities across signals. A key part of our argument is that complementarities ``wash out" as information is acquired, so that  signals are eventually perceived as  providing (approximately) independent information. After sufficiently many observations, the best next signal to acquire is (generically) a part of the best next pair of signals to acquire (and the best batch of $B$ signals, where $B$ is sufficiently large, is always a part of the best batch of $2B$ signals).

We now proceed with a more detailed intuition.

\bigskip

Consider a one-shot version of our problem, in which the DM allocates $t$ observations across the available signals. Define a \emph{$t$-optimal} ``division vector'' $n(t)$ to be any optimal allocation of these signals:
\[
n(t) =(n_1(t), \dots, n_K(t))\in \argmin_{(q_1,\dots,q_K) : q_i \in \mathbb{Z}_+, \sum_{i=1}^K q_i=t} f(q_1, \dots, q_K)
\]
where $n_i(t)$ is the number of observations allocated to signal $i$. Applying \citet{HansenTorgersen}, this allocation maximizes expected payoffs for any decision.

We study the evolution of the vectors $n(t)$ as  the number of observations $t$ varies. If each count $n_i(t)$ increases monotonically in $t$, then the division vectors $(n(t))_{t\geq 1}$ can be achieved under some sequential sampling rule; moreover, this sampling rule corresponds to myopic information acquisition.\footnote{Proceed by induction: in the first period the myopic rule chooses the signal that minimizes posterior variance. In the second period, he again wants to minimize posterior variance at the given period; since the division chosen by the totally optimal rule is best and feasible given the period $1$ choice, this is what myopic information acquisition will yield. So on and so forth.} A key ``dynamic Blackwell" lemma shows that a sequence of normal signals is better than another sequence (for all decision problems depending on $\theta_1$) if and only if it leads to lower posterior variances at every period.\footnote{This generalizes a result from \cite{Greenshtein}, see Section \ref{relatedLit} for further discussion.} Thus, optimality of myopic information acquisition directly follows from existence of a sequence $(n(t))_{t\geq1}$ of monotonically increasing division vectors.

Existence of such a sequence depends on whether there are strong complementarities across signals. In Example 2 of Figure \ref{fig:intuition}, signals are ``independent," and any allocation that is as close to balanced across the signals as possible is $t$-optimal. It is thus possible to find $t$-optimal division vectors that evolve monotonically. Theorem \ref{propSeparable} generalizes Example 2 to a \emph{class} of environments in which $n(t)$ evolve monotonically, and the optimal rule is myopic from period 1. 

\begin{figure}[h] 
\begin{center}
\includegraphics[scale=0.5]{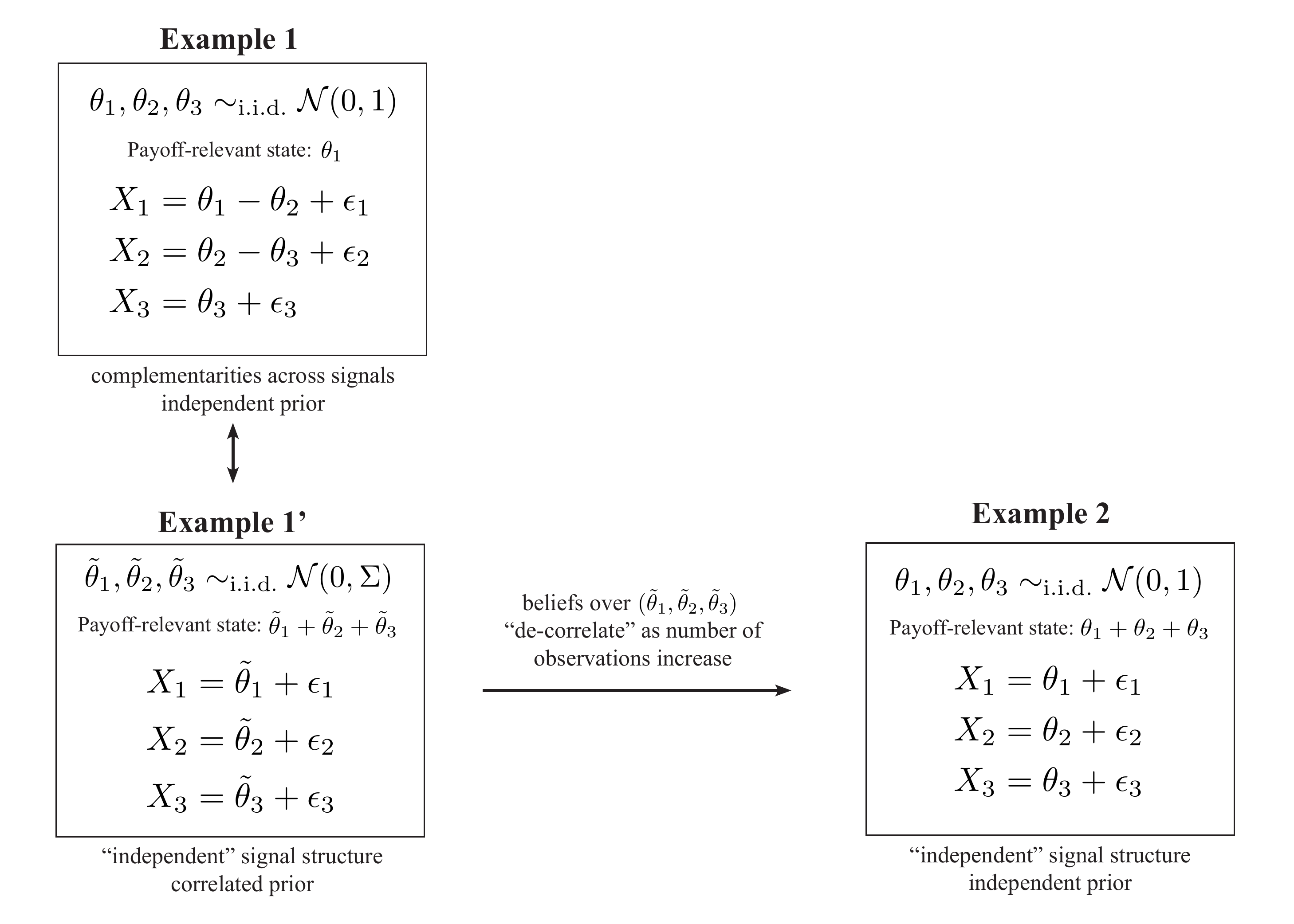}
\caption{}\label{fig:intuition}
\end{center}
\end{figure}
In general, the division vectors $n(t)$ need not be monotone, as shown in Example 1 of Figure \ref{fig:intuition}. Because signals $X_2$ and $X_3$ are strong complements (observation of either increases the value of observation of the other), we have that $n(5)=(4,1,0)$ but $n(6)=(3,2,1)$; that is, after four signal acquisitions, the best next signal to acquire is $X_1$, but the best pair of signals to acquire is $\{X_2,X_3\}$. The optimal allocation of six signals is not achievable from the best allocation of five signals.\footnote{To see these are the $t$-optimal divisions, we calculate that at $t = 5$, $f(4,1,0) = \frac{11}{23} < \frac{14}{29} = f(3,1,1) = f(3,2,0)$. Whereas at $t = 6$, $f(3,2,1) = \frac{5}{11} < \frac{17}{37} = f(4,1,1) = f(4,2,0)$.} 

An important step in our proofs is to show that environments which start off like Example 1 necessarily become ``like" Example 2 (as observations of each signal accumulate). To facilitate this comparison, rewrite Example 1 in the following way: Define a new set of states $\tilde{\theta}_1=\theta_1-\theta_2$, $\tilde{\theta}_2=\theta_2-\theta_3$, and $\tilde{\theta}_3 = \theta_3$. The original prior over $(\theta_1,\theta_2,\theta_3)$ defines a new prior over $(\tilde{\theta}_1,\tilde{\theta}_2,\tilde{\theta}_3)$, and the original payoff-relevant state can be re-expressed in the new states as $\theta_1= \tilde{\theta}_1+\tilde{\theta}_2+\tilde{\theta}_3$. The signal structure in Example 1' is the same as in Example 2, with the crucial difference that the prior is correlated in Example 1' and independent in Example 2. Given sufficiently many observations of each signal in Example 1', however, the DM's posterior beliefs over $(\tilde{\theta}_1, \dots, \tilde{\theta}_K)$ become \emph{almost independent}. That is, not only does learning about each $\tilde{\theta}_k$ occur, but the states $\tilde{\theta}_k$ ``de-correlate." Thus, eventually the signals can be viewed as approximately independent.

The above heuristic statements about de-correlation roughly follow from a Bayesian version of the Central Limit Theorem. We establish a technical lemma (Lemma \ref{lemmSecondDer} in Appendix \ref{appxPrelim}) that strengthens this with a comparison of the value of different signals. Specifically, we characterize the externality that observation of a given signal $X_i$ has on the marginal value of future observations. We show that the effect of observing $X_i$ on the value of future observations of $X_i$ (eventually) far outweighs its effect on future observations of any $X_j$, $j\neq i$. That is, the only effect that observation of  $X_i$ can have on the ranking of signals is by reducing the position of signal $i$. This property is immediate when the covariance matrix is the identity (as in Example 2), and holds also when the covariance matrix is close to diagonal (almost independent). At late periods we have a setting much like Example 2, and the problem is ``near-separable."

Now finally, observe that the transformation we used to rewrite Example 1 as Example 1' was not special. Indeed, we can rewrite any signal structure in the following way: For each signal $k$, define a new state $\tilde{\theta}_k=\langle c_k,\theta\rangle$, so that the signal $X_k$ is simply $\tilde{\theta}_k$ plus independent Gaussian noise. Under Assumption \ref{assumptionIdentifiability} (no redundant signals), the payoff-relevant state can be rewritten as a unique linear combination of the transformed states $\tilde{\theta}_1, \dots, \tilde{\theta}_K$, and the original prior defines a new prior over $(\tilde{\theta}_1,\dots,\tilde{\theta}_K)$. The same assumption allows us to show each signal is sampled infinitely often. Hence ``de-correlation" necessarily occurs.

The arguments above form the core of our proofs of Theorem \ref{propBatch} and Theorem \ref{propGenericEventual}. They establish that $n(t)$ eventually evolves approximately monotonically.\footnote{See Appendix \ref{appxExample2} for an example in which $n(t)$ fails to be monotone even when we restrict to arbitrarily late periods.} We introduce two different conditions that allow us to strengthen this to (eventual) \emph{exact} optimality of myopic information acquisition.

Our first approach is to allow for acquisition of a larger number of signals each period. We show that in transitioning from the $t$-optimal division to the $(t+d)$-optimal division (where $t$ is sufficiently large relative to $d$), if some signal count decreases (failing sequentiality), then every other signal count increases by at most one. Specifically, taking $d = K-1$, we can show that $n_i(t+K-1) \geq n_i(t)$ for every signal $i$ at every period $t \geq T$, with $T$ a large number depending on the informational environment. Thus, given block size $B \geq \max\{K-1, T\}$, the division vectors $n(Bt)$ are attainable using a sequential rule from period $1$. Applying the dynamic Blackwell lemma mentioned earlier, we have that the optimal strategy immediately follows $n(Bt)$, myopically. This is the intuition behind Theorem \ref{propBatch}.

Our second approach quantifies the ``typicality" of failures of monotonicity of $n(t)$. (Below, assume $B=1$ for illustration.) We show  that at late periods $t$, these failures occur for a purely technical reason: The vectors $n(t)$ eventually seek to approximate an optimal limiting frequency $\lambda \in \Delta^{K-1}$ across signals, but do so under an integer constraint; and the best integer approximations to $\lambda t$ may not be monotone.\footnote{This is indeed the case for Example 2, see Appendix \ref{appxExample2} for details.} A key lemma demonstrates these integer approximations do (eventually) evolve monotonically for a ``measure-1" set of coefficient matrices $C$.\footnote{The lemma follows from results in Diophantine approximation theory, which studies the extent to which an arbitrary real number can be approximated by rational numbers.} Thus, if the optimal strategy coincides with $n(t)$ at some late period $t$, it will follow these $t$-optimal divisions afterwards.

The last step is to verify that the optimal strategy will coincide with $n(t)$. We argue that if this were not true, then the DM could ``deviate toward the $t$-optimal divisions" and (in generic environments) improve the posterior variance at every period,  contradicting optimality of the original strategy. Hence generically, the optimal strategy will eventually coincide with $n(t)$ and subsequently follow it. This yields Theorem \ref{propGenericEventual}.

\subsection{Precision vs.\ Correlation} \label{sec:precisionCorrelation}
With sufficiently many observations, the DM's beliefs simultaneously become more precise and less correlated, and these two effects are confounded in our main results. It is tempting to think that Theorem \ref{propBatch} (or Theorem \ref{propGenericEventual}) follows from the eventual precision of beliefs. However, as our discussion above suggests, the key feature is not how precise beliefs are, but how correlated they are. Specifically, the block size $B$ needed in Theorem \ref{propBatch} depends on how many observations are required for the transformed states $\tilde{\theta}_1, \dots, \tilde{\theta}_K$ to ``de-correlate,'' at which point complementarities across signals are weak. 

Below we make this formal with a bound on $B$. To state our result, we define transformed states $\tilde{\theta}_k = \frac{1}{\sigma_k}\langle c_k, \theta \rangle$ (dividing through by $\sigma_k$ normalizes all error variances to $1$), and let $\tilde{V}$ denote the prior covariance matrix over these transformed states. The payoff-relevant state $\theta_1$ can be rewritten as a linear combination of the transformed states: $\theta_1 = \langle w, \tilde{\theta} \rangle$ for some fixed payoff weight vector $w \in \mathbb{R}^{K}$. In the following result we assume for simplicity that $w = \mathbf{1}$ is the vector of ones, although our analysis can be easily adapted to any $w$. 

\begin{proposition}\label{propBoundOnB}
Let $R$ denote the operator norm of the matrix $(\tilde{V})^{-1}$.\footnote{The operator norm of a matrix $M$ is defined as $\Vert M \Vert_{op} = \sup \left\{ \frac{\Vert Mx \Vert }{\Vert x \Vert} \, : \, x \in \mathbb{R}^K \mbox{ with } x\neq \bold{0}\right\}$.} Suppose $w = \mathbf{1}$, then $B \geq 8(R+1)K^{1.5}$ is sufficient for Theorem \ref{propBatch} to hold. 
\end{proposition}

Notice that this bound is increasing in the norm of $\tilde{V}^{-1}$. To interpret this, suppose first that we adjust the precision of the DM's prior beliefs over $\tilde{\theta_1}, \dots, \tilde{\theta}_K$ but fix the degree of correlation, for example by scaling $\tilde{V}$ by a factor less than $1$. Then the norm of $\tilde{V}^{-1}$ increases, and a larger number of signals $B$ is needed. This is because a more precise prior can be understood as ``re-scaling" the state space by shrinking all states towards zero. Since signal noise is not correspondingly rescaled, each signal now reveals less about the states, and de-correlation takes longer. 

In contrast, suppose we hold prior precision fixed and increase the degree of prior correlation. This would correspond to fixing the diagonal entries of $\tilde{V}$ and increasing the off-diagonal entries, so that the variances about individual states are unchanged but their covariances become larger in magnitude. Then the entire matrix $\tilde{V}$ becomes closer to being singular, the norm of its inverse increases and a larger $B$ is required. That is, a more correlated prior requires more observations to de-correlate.

Finally, we highlight that Proposition \ref{propBoundOnB} only provides a sufficient condition on the block size for the myopic rule to be optimal. However, the aforementioned comparative statics hold not just for the $B$ that we identify, but also for the smallest $B$ that produces the result. Indeed, these comparative statics are sharp in the continuous-time limit of our model (see Section \ref{ex:contTime}), where doubling the prior precision would also double the capacity $B$ needed for Theorem \ref{propBatch} to be true.  

To summarize, optimality of myopic information acquisition obtains quickly when prior beliefs are \emph{imprecise} and \emph{not too correlated}. Our results suggest that despite the amount of (residual) uncertainty in these situations, there is not much conflict between short-run and long-run information acquisition incentives. 

\subsection{How Important is Normality?} \label{sec:normality}

\emph{De-correlation.}  The key part of our argument is that signals eventually de-correlate.  This de-correlation derives from a Bayesian version of the Central Limit Theorem, and does not rely on special properties of normality. Consider a more general setting with signals $X_i = \tilde{\theta}_i + \epsilon_i$, where the noise term $\epsilon_i$ has an arbitrary distribution with zero mean and finite variance. Then, we have that the suitably normalized posterior distribution over $(\tilde{\theta}_1,\dots, \tilde{\theta}_K)$ converges towards a standard normal distribution (so that signals are approximately independent). We thus expect that given sufficiently many observations, our previous comparisons on the value of information extend beyond normal signals. 

\smallskip

But if we drop normality, then our results weaken in the following ways. Specifically, in working with general signal distributions, we sacrifice the potential for immediate optimality of the myopic rule, and also the generality to all intertemporal decisions.

\smallskip

\emph{Immediate Optimality of the Myopic Rule.} For normal signals, we established a $T$ such that given $T$ observations of each signal, the posterior covariance matrix (over the transformed states) is almost independent. Notably, this bound holds uniformly across all histories of signal realizations, thanks to the fact that posterior variances do not depend on signal realizations under normality. As mentioned above, we can use a Bayesian Central Limit Theorem to argue a similar property for other signal distributions. The difference is that the CLT gives us (near) independence 
\emph{almost surely}, so that at every period $t$, there is still positive probability (albeit vanishing as $t$ increases) that the normalized posterior covariance matrix is very different from the identity. This precludes us from demonstrating a block size $B$ given which the optimal rule would be myopic \emph{from period 1} (Theorem  \ref{propBatch}). For general signal distributions, we thus conjecture that \emph{almost surely} the optimal rule is eventually myopic, but do not know what conditions would produce immediate optimality of the myopic rule. 

\smallskip

\emph{General Intertemporal Payoffs.} The place where we rely most heavily on normality is the statement that our results hold for \emph{all} payoff criteria that depend only on $\theta_1$ (and actions). Indeed, when payoff-relevant uncertainty is one-dimensional (as it is here), then all normal signals can be Blackwell-ordered based on their posterior variances. We use this fact in Section \ref{sec:defineMyopic} when we define the myopic rule to maximize reduction in posterior variance. We use this fact again in Section \ref{sec:intuition} when we define the $t$-optimal divisions $n(t)$ without explicit reference to the payoff function.

Finally, while the above uses of normality are concerned with static decisions (i.e.\ taking an action once), we also need normality to be able to compare \emph{signal sequences}. Generalizing \cite{Greenshtein}, we show that the ranking of sequences of normal signals is the same whether we consider the class of static decision problems or the broader class of intertemporal decisions. This equivalence does not hold in general; see \cite{Greenshtein} for a counterexample involving Bernoulli signals.

\section{Games with Dynamic Information Acquisition} \label{games}

We now extend our results to a multi-player setting in which individuals privately acquire information before playing a normal-form game at a random (and exogenously determined) end date.\footnote{\cite{Reinganum} considers a similar multi-agent model with private information acquisition (specifically, firms engaging in R\&D before competing in oligopoly). Her model is further developed by \cite{Taylor} within the context of research tournaments. However, these papers assume perfectly revealing signals and are thus distinguished from our setting.} 

There are $N$ players (indexed by $i$), each of whom has access to a set of $K$ signals 
\[
X^i_k= \langle c^i_k, \bold{\theta}^{i} \rangle + \epsilon^i_k.
\]
The state vector $\theta^i = (\theta^i_1, \theta^i_2, \dots, \theta^i_K)$ represents persistent unknown states particular to player $i$. Noise terms $\epsilon_k^{it}$ are (normalized to) standard normal  and independent across signals, players and time. 

In each period up to and including the final period, each player $i$ acquires $B$ independent observations of his signals described above, possibly obtaining multiple (independent) realizations of the same signal. Signal choices and their realizations are both \emph{privately observed}. The final period is determined by an exogenous distribution $\pi$. At this period, agents play a one-shot incomplete information game, where each player $i$'s payoff $u_i(a,\omega)$ depends on actions $a = (a^1, \dots, a^N)$ in addition to a payoff-relevant state $\omega$. 

We require that the players share a common prior over all states ($\omega$ and the player-specific state vectors $(\theta^i)_{1 \leq i \leq N}$) with the following \emph{conditional independence} property: For each player $i$, conditional on the value of $\theta^i_1$, both the payoff-relevant state $\omega$ and also the \emph{other} players' unknown states $(\theta^j)_{j \neq i}$ are conditionally independent from player $i$'s own states $\theta^i$.\footnote{Note that we do not impose conditional independence \emph{between} $\omega$ and the other players' states.} This ensures that no player $i$ infers anything about $\omega$ or about any other player $j$'s information beyond what he (player $i$) learns about $\theta^i_1$, which makes $\theta^i_1$ the only state of interest for player $i$.\footnote{Conditional independence is imposed on players' prior beliefs. However, this implies conditional independence for subsequent posterior beliefs; given the value of $\theta^i_1$, each signal is simply a linear combination of player $i$'s other states plus noise, 
thus conditional independence is preserved under updating.}

For concreteness, we provide examples (adapted from \cite{LambertOstrovskyPanov}) that do and do not satisfy conditional independence.\footnote{Example \ref{ex:satisfiesCI} is based on Example OA.3 in \cite{LambertOstrovskyPanov}. Example \ref{ex:failsCI} is based on their Example 1.}

\begin{example}[Satisfies Conditional Independence] \label{ex:satisfiesCI} In addition to the payoff-relevant state $\omega$, there is a common unknown state $\xi$, and two player-specific unknown states $b_1$ and $b_2$. All states are independent. Player 1 has access to  observations of $\omega + \rho_1 \xi + b_1 +\epsilon_1^1$ (where $\rho_1$ is a constant) and  $b_1+\epsilon_2^1$. Player 2 has access to observations of $\omega + \rho_2 \xi + b_2+\epsilon_1^2$ (where $\rho_2$ is a constant) and $b_2+\epsilon_2^2$. To see that this example satisfies Conditional Independence, define $\theta_1^1 = \omega + \rho_1\xi$ and $\theta^2_1=\omega + \rho_2 \xi$.
\end{example}

\begin{example}[Fails Conditional Independence] \label{ex:failsCI} The payoff-relevant state is $\omega$, and there is additionally a common unknown state $\xi$. These states are independent. Player 1 has access to noisy observations of $\omega + \xi$ only. Player 2 has access to noisy observations of both $\omega$ and $\xi$. Because both states $\omega$ and $\xi$ covary with $\omega + \xi$, there is no way to define the second player's ``state of interest" $\theta_1^2$ that would satisfy Conditional Independence. 
\end{example}

We maintain Assumption \ref{assumptionIdentifiability}, so that signals are non-redundant (players must observe all the signals available to them in order to learn $\theta^i_1$). The following result generalizes Theorem \ref{propBatch} and Theorem \ref{propSeparable}. (Although we do not state it here, Theorem \ref{propGenericEventual} also extends.)

\begin{corollary} \label{corDominant} 
Suppose $B$ is sufficiently large or the informational environment is separable. Then there exists a Nash equilibrium of this model where each player acquires information myopically.
\end{corollary}

\noindent In fact, we show that the myopic information acquisition strategy is \emph{dominant} in the following sense: For arbitrary opponent strategies, player $i$'s best response involves acquiring signals myopically. 

In Appendix \ref{appxGames}, we apply the above corollary to extend results from \cite{HellwigVeldkamp} and \cite{LambertOstrovskyPanov} to a setting with sequential information acquisition.

\section{Extensions} \label{extensions}

\subsection{Endogenous Learning Intensities}\label{ex:intensity}
The main model imposes an exogenous capacity constraint of $B$ signals per period. Suppose now that in each period $t$, the DM can choose to observe any number $N_t \in \mathbb{Z}_{+}$ of signal realizations (which are then optimally allocated across signals). The DM incurs a flow cost of information acquisition, modeled as $\kappa(N_t)$ for some increasing cost function $\kappa(\cdot)$ with $\kappa(0) = 0$. This framework embeds our main model if we define $\kappa(N) = 0$ for $N \leq B$ and $\kappa(N) = \infty$ for $N > B$. 

We assume that the DM's payoff is $U(a_1, a_2, \dots; \theta_1) - \sum_{t} \delta^{t-1} \cdot \kappa(N_t)$ for some discount factor $\delta$.\footnote{Our analysis can accommodate more general payoff functions of the form $U(N_1, a_1, N_2, a_2, \dots; \theta_1)$.} For the special case of endogenous stopping, the payoff function simplifies to
\[
\delta^{\tau} \cdot u(a_\tau; \theta_1) - \sum_{t = 1}^{\tau} \delta^{t-1} \cdot \kappa(N_t)
\]
whenever the DM stops after $\tau$ periods. This is a discrete-time generalization of the framework proposed in \cite{MoscariniSmith}, although our focus is on allocation of the signals instead of choice of intensity level.\footnote{\citet{MoscariniSmith} has a single state and a single signal $(K = 1)$, so the DM chooses only the learning intensities $N_t$. Unlike \citet{MoscariniSmith}, we do not characterize the optimal sequence of intensity choices $(N_t)_{t\geq 1}$, but instead show how this problem can be separated from allocation of those observations across different kinds of sources.} 

Theorems \ref{propBatch} and \ref{propSeparable} generalize to this setting: 
\begin{corollary} \label{corIntensity} 
Suppose $B$ is sufficiently large or the informational environment is separable. Then, even with endogenous learning intensities, the DM has an optimal strategy that chooses signals myopically.
\end{corollary}

In the above corollary, ``myopic acquisition" means the following: In any period $t$, given (endogenous) intensity choice $N_t$, the optimal acquisitions are the $N_t$ signals that minimize posterior variance about $\theta_1$. We emphasize that while myopic signal choices are optimal, myopic intensity choices need not be. However, knowing that the signal choices must follow the myopic path provides a simplifying first step towards the characterization of optimal intensity levels.

Generic eventual myopia (Theorem \ref{propGenericEventual}) also extends, but we omit the details.

\subsection{Multiple Payoff-Relevant States} \label{ex:multiStates}
In the main model, the DM's payoff function depends on a one-dimensional state. Our results do not extend in general to payoff functions that depend on the full state vector $(\theta_1,\dots,\theta_K)$. Loosely, this is because the signals can no longer be Blackwell ordered; thus, even the statement that the signal which maximally reduces posterior variance is best for \emph{static} decision problems has no analogue when multiple states are payoff-relevant. 

However, our results do extend for a class of prediction problems. Specifically, suppose that at an exogenous end date (determined by an arbitrary distribution over periods), the DM is asked to predict the state vector $\theta$. At this time, he receives a payoff of 
\[
-(a- \theta)'\mathbf{W}(a-\theta),
\]
where $\mathbf{W}$ is a given positive semi-definite matrix and $a \in \mathbb{R}^K$ is the DM's prediction. 

In such a setting, our main results and their proofs extend essentially without modification. To see this, note that when $\mathbf{W}$ is diagonal, the DM simply minimizes a \emph{weighted sum} of posterior variances about multiple states. Generalizing Lemma \ref{lemmSecondDer} in Appendix \ref{appxPrelim}, we can show that any such objective function exhibits ``eventual near-separability," which is sufficient to derive Theorem \ref{propBatch} and Theorem \ref{propGenericEventual}. When $\mathbf{W}$ is not diagonal, we can use the spectral theorem to write the DM's objective function as a weighted sum of posterior variances about some linearly-transformed states. Our proofs still carry through. 

\subsection{Continuous Time} \label{ex:contTime}
In a working paper, we analyze a continuous-time version of our problem. In that model, the DM has $B$ units of attention in total at every point in time. He chooses attention levels $\beta_1(t), \dots, \beta_K(t)$ subject to $\beta_i(t) \geq 0$ and $\sum_{i} \beta_i(t) \leq B$, and then observes $K$ diffusion processes $X_1, \dots, X_K$, whose evolutions are affected by the attention rates in the following way:
\[
dX_i(t) = \beta_i(t) \cdot \langle c_i, \theta \rangle ~dt + \sqrt{\beta_i(t)} ~d B_i,
\]
where each $B_i$ is an independent standard Brownian motion. This formulation can be seen as a limit of our discrete-time model in the current paper, where we take period length to zero and also ``divide" the signals to hold constant the amount of information that can be gathered every instant.

 In short, all results from this paper extend (and occasionally can be strengthened): the optimal rule is eventually myopic in \emph{all} informational environments (thus dropping the generic qualifier in Theorem \ref{propGenericEventual}); additionally, we provide more permissive sufficient conditions on the informational environment under which an optimal strategy is myopic from period $1$. 
We refer the reader to the working paper for more detail.

\section{Related Literature} \label{relatedLit}

Besides the references mentioned in the introduction, our setting is related to a recent literature \citep{Bubeck2009, Russo} regarding ``best-arm identification" in a multi-armed bandit setting: A DM samples for a number of periods before selecting an arm and receiving its payoff. In Appendix \ref{appxK=2}, we characterize the optimal information acquisition strategy for the case of two states ($K=2$), which exactly applies to the problem of identifying the better of two correlated normal arms. However, due to our assumption of an one-dimensional payoff-relevant state, we are not able to handle more than two arms.\footnote{With two arms, the DM only cares about the difference in their expected payoffs. 
Choosing among more than two arms would involve \emph{multi-dimensional payoff uncertainty} and \emph{a decision problem that is not prediction}. As we discussed in \ref{ex:multiStates}, the lack of a complete Blackwell ordering limits the generalization of our argument. Incidentally, in related sequential search settings, \cite{Sanjurjo}, \cite{KeVillas-Boas} and \cite{ChickFrazier} also highlight the challenge of characterizing the optimal strategy once there are at least three alternatives.} 

We note that correlation is the key feature of our setting, and are not aware of many papers that study correlated bandits, either under the classical framework or under best-arm identification (see \cite{Rothschild}, \cite{Keener} and \cite{Mersereau2009} for a few stylized cases). 

Our results on the comparison of sequential normal experiments (see the discussion in Section \ref{discussion}, and results in Appendix \ref{appxBlackwell}) generalize the main result in \cite{Greenshtein}. \cite{Greenshtein} compares two \emph{deterministic (i.e.\ history-independent)} sequences of signals, where each signal is $\theta_1$ plus \emph{independent} normal noise. His Theorem 3.1 implies that the former sequence is Blackwell-dominant if and only if its cumulative precision is higher at every time. Note that this statement does not refer to the prior beliefs, but if we impose a normal prior on $\theta_1$, then higher cumulative precision is equivalent to lower posterior variance. Thus, the result of \cite{Greenshtein} coincides with ours when $\theta_1$ is the only persistent state, and when all signals are independent conditional on $\theta_1$. Our setting features additional correlation across different signals through the persistent (payoff-irrelevant) states $\theta_2, \dots, \theta_K$. Consequently, the dynamic Blackwell comparison in our model depends on prior beliefs.\footnote{This is already the case for static comparisons, since as the prior beliefs vary, it is not always the same signal that leads to the lowest posterior variance about $\theta_1$.} This feature, together with the endogenous choice of signals (which may be history-dependent), complicates our problem relative to \cite{Greenshtein}.

Finally, our work is closely related to \emph{optimal design}, a field initiated by the the early work of \cite{Robbins1952} (see \cite{Chernoff1972} for a survey). Specifically, the problem of one-shot allocation of $t$ signals (our $t$-optimal criterion in Section \ref{discussion}) is equivalent to a Bayesian optimal design problem with respect to the ``$c$-optimality criterion'', which seeks to minimize the variance of an unknown parameter. Our analysis is however focused on dynamics, and we demonstrate here the optimality of ``greedy design" for a broad class of (intertemporal) objectives.

\section{Conclusion} \label{conclusion}

A DM learns about a payoff-relevant state by sequentially sampling (batches of) signals from flexibly correlated Gaussian sources. Under conditions that we provide, myopic information acquisition is optimal and robust across all possible payoff functions. Generically, the optimal strategy eventually acquires signals myopically. These results are robust to extension to multi-player settings, to endogenous choice of the number of signals to acquire each period, and to multi-dimensional uncertainty for certain payoff functions. 

We conclude with a re-interpretation of the main setting. Suppose there is a sequence of short-lived decision-makers indexed by time $t$, each of whom acquires information and then takes an action $a_t$ to maximize some private objective $u_t(a_t,\theta_1)$. All information acquisition is public.\footnote{This separates our model from classic social learning frameworks \citep{Banerjee,Hirshleifer}, where decision-makers only observe coarse summary statistics of past information acquisitions.} A social planner has (an arbitrary) objective function $U(a_1,a_2,\dots;\theta_1)$; thus, his incentives are misaligned with those of the short-lived decision makers. Our main results demonstrate conditions under which this mis-alignment is of no consequence. If each DM acquires sufficiently many signals, or if the environment is separable, then each DM will acquire exactly the information that the social planner would have wanted. Generically, the social planner will not be able to improve (at late periods) upon the information that has been aggregated so far. We generalize this qualitative insight in our companion piece \citet{LiangMu2017} and demonstrate also how it can fail.

\clearpage

\input{appendix}

\clearpage

\input{onlineappendix}

\end{document}

%% file: appendix.tex
\section{Appendix}

\subsection{Preliminary Results} \label{appxPrelim}

\subsubsection{Posterior Variance Function}
We begin by presenting basic results that are used throughout the appendix. The following lemma characterizes the posterior variance function $f$ mentioned in the main text, which maps signal counts to the DM's posterior variance about the payoff-relevant state $\theta_1$. 

\begin{lemma}\label{lemmVar}
Given prior covariance matrix $V^0$ and $q_i$ observations of each signal $i$, the DM's posterior variance about $\theta_1$ is given by\footnote{When $M$ is a matrix, we let $M_{ij}$ denote its $(i,j)$-th entry.}
\begin{equation}\label{eq:f}
f(q_1, \dots, q_K) = [V^0- V^0 C'\Sigma^{-1} CV^0]_{11} 
\end{equation}
where $\Sigma = CV^0C' + D^{-1}$ and $D = \diag\left(\frac{q_1}{\sigma_1^2}, \dots, \frac{q_K}{\sigma_K^2}\right)$. The function $f$ is decreasing and convex in each $q_i$ whenever these arguments take non-negative extended real values: $q_i \in \overline{\mathbb{R}_{+}} = \mathbb{R}_{+} \cup \{+\infty\}$.
\end{lemma}

\begin{proof}
The expression (\ref{eq:f}) comes directly from the conditional variance formula for multivariate Gaussian distributions. To prove $\frac{\partial f}{\partial q_i} \leq 0$, consider the partial order $\succeq$ on positive semi-definite matrices so that $A \succeq B$ if and only if $A - B$ is positive semi-definite. As $q_i$ increases, the matrices $D^{-1}$ and $\Sigma$ decrease in this order. Thus $\Sigma^{-1}$ increases in this order, which implies that $V^0- V^0 C'\Sigma^{-1} CV^0$ decreases in this order. In particular, the diagonal entries of $V^0- V^0 C'\Sigma^{-1} CV^0$ are uniformly smaller, so that $f$ becomes smaller. Intuitively, more information always improves the decision-maker's estimates. 

To prove $f$ is convex, it suffices to prove $f$ is \textit{midpoint-convex} since the function is clearly continuous. Take $q_1, \dots, q_K$, $r_1, \dots, r_K \in \overline{\mathbb{R}_{+}}$ and let $s_i = \frac{q_i + r_i}{2}$.\footnote{We allow the function $f$ to take $+\infty$ as arguments. This relaxation does not affect the properties of $f$, and it is convenient for our future analysis.} Define the corresponding diagonal matrices to be $D_q$, $D_r$, $D_s$. We need to show $f(q_1, \dots, q_K) + f(r_1, \dots, r_K) \geq 2f(s_1, \dots, s_K)$. For this, we first use the Woodbury inversion formula to write
\[ \Sigma^{-1} = (CV^0C' + D^{-1})^{-1} = J - J(J+D)^{-1}J, \]
with $J = (CV^0C')^{-1}$. Plugging this back into (\ref{eq:f}), we see that it suffices to show the following matrix order:
\[\frac{(J+D_q)^{-1} + (J+D_r)^{-1}}{2} \succeq (J+D_s)^{-1}. \]
Inverting both sides, we need to show
$2 \left((J+D_q)^{-1} + (J+D_r)^{-1}\right)^{-1} \preceq J+D_s. $
From definition, $D_q+ D_r = \diag(\frac{q_1+r_1}{\sigma_{1}^2}, \dots, \frac{q_K+r_K}{\sigma_{K}^2}) = 2D_s.$ Thus the above follows from the AM-HM inequality for positive definite matrices, see for instance \cite{Ando}.
\end{proof}

\subsubsection{The Matrix $Q_i$}
Let us define for each $1 \leq i \leq K$, 
\begin{equation}\label{eq:Q}
Q_i = C^{-1}\Delta_{ii}C'^{-1}
\end{equation}
where $\Delta_{ii}$ is the matrix with `1' in the $(i,i)$-th entry, and zeros elsewhere. We note that $[Q_i]_{11} = ([C^{-1}]_{1i})^2$, which is strictly positive under Assumption \ref{assumptionIdentifiability}. These matrices $Q_i$ will be repeatedly used in our proofs. 

\subsubsection{Order Difference Lemma}
Here we establish the asymptotic order for the second derivatives of $f$. 
\begin{lemma}\label{lemmSecondDer} 
As $q_1, \dots, q_K \to \infty$, $\frac{\partial^2 f}{\partial q_i^2}$ is positive with order $\frac{1}{q_i^3}$, whereas $\frac{\partial^2 f}{\partial q_i \partial q_j}$ has order at most $\frac{1}{q_i^2q_j^2}$ for any $j \neq i$. Formally, there is a positive constant $L$ depending on the informational environment, such that $\frac{\partial^2 f}{\partial q_i^2} \geq \frac{1}{Lq_i^3}$ and $\lvert \frac{\partial^2 f}{\partial q_i \partial q_j} \rvert \leq \frac{L}{q_i^2q_j^2}$. 
\end{lemma}

\noindent To interpret, the second derivative $\partial^2 f/\partial q_i^2$ is the effect of observing signal $i$ on the marginal value of the next observation of signal $i$. Our lemma says that this second derivative is always eventually positive, so that each observation of signal $i$ makes the next observation of signal $i$ less valuable. The cross-partial $\partial^2 f/\partial q_i \partial q_j$ is the effect of observing signal $i$ on the marginal value of the next observation of a different signal $j$, and its sign is ambiguous. 

The key content of the lemma is that regardless of the sign of the cross partial, it is always of lower order compared to the second derivative. In words, the effect of observing a signal on the marginal value of other signals (as quantified by the cross-partial) is eventually second-order to its effect on the marginal value of further realizations of the same signal (as quantified by the second derivative). This is true for any signal path in which the signal counts $q_1, \dots, q_K$ go to infinity proportionally, which we will justify later.

\begin{proof}
Recall from Lemma \ref{lemmVar} that
$
f(q_1, \dots, q_K) = [V^0- V^0 C'\Sigma^{-1} CV^0]_{11}
$
and therefore
\begin{equation}\label{eq:der1}
\frac{\partial^2 f}{\partial q_i \partial q_j} = [\partial_{ij}(V^0 - V^0C'\Sigma^{-1} CV^0)]_{11} \quad \quad \quad \quad \frac{\partial^2 f}{\partial q_i^2} = [\partial_{ii}(V^0 - V^0C'\Sigma^{-1} CV^0)]_{11}
\end{equation}

\noindent Using properties of matrix derivatives,
\begin{equation}\label{eq:der2}
\partial_{ii} (\Sigma^{-1}) = \Sigma^{-1} (\partial_i \Sigma) \Sigma^{-1} (\partial_i \Sigma) \Sigma^{-1} -  \Sigma^{-1} (\partial_{ii} \Sigma ) \Sigma^{-1} +  \Sigma^{-1} (\partial_i  \Sigma)\Sigma^{-1} (\partial_i \Sigma) \Sigma^{-1}.
\end{equation}
The relevant derivatives of the covariance matrix $\Sigma$ are
\[
\partial_{i} \Sigma = -\frac{\sigma_{i}^2}{q_i^2} \Delta_{ii} \quad \quad \quad \quad \partial_{ii} \Sigma = \frac{2\sigma_{i}^2}{q_i^3} \Delta_{ii} 
\]

\noindent Plugging these into (\ref{eq:der2}), we obtain 
$
\partial_{ii} (\Sigma^{-1}) = -\frac{2\sigma_{i}^2}{q_i^3} (\Sigma^{-1} \Delta_{ii} \Sigma^{-1}) + O\left(\frac{1}{q_i^4}\right). 
$
Thus by (\ref{eq:der1}), 
\begin{equation}\label{eq:der3}
\frac{\partial^2 f}{\partial q_i^2} = \left[- V^0C' \cdot \frac{\partial^2 (\Sigma^{-1})}{\partial q_i^2} \cdot CV^0\right]_{11}  = \frac{2\sigma_{i}^2}{q_i^3} \cdot \left[V^0C' \Sigma^{-1} \Delta_{ii} \Sigma^{-1} CV^0\right]_{11} + O\left(\frac{1}{q_i^4}\right).
\end{equation}
As $q_1, \dots, q_k \to \infty$, $\Sigma \to CV^0C'$ which is symmetric and non-singular. Thus the matrix $V^0C' \Sigma^{-1} \Delta_{ii} \Sigma^{-1} CV^0$ converges to the matrix $Q_i$ defined earlier in (\ref{eq:Q}). From (\ref{eq:der3}) and $[Q_i]_{11} > 0$, we conclude that $\frac{\partial^2 f}{\partial q_i^2}$ is positive with order $\frac{1}{q_i^3}$. Similarly, for $i \neq j$, we have 
\[
\partial_{ij} (\Sigma^{-1}) = \Sigma^{-1} (\partial_j \Sigma) \Sigma^{-1} (\partial_i \Sigma) \Sigma^{-1} -  \Sigma^{-1} (\partial_{ij} \Sigma ) \Sigma^{-1} +  \Sigma^{-1} (\partial_i  \Sigma)\Sigma^{-1} (\partial_j \Sigma) \Sigma^{-1}.
\]
The relevant derivatives of the covariance matrix $\Sigma$ are 
\[
\partial_{i} \Sigma = -\frac{\sigma_{i}^2}{q_i^2} \Delta_{ii} \quad \quad  \partial_{j} \Sigma = -\frac{\sigma_{j}^2}{q_j^2} \Delta_{jj} \quad \quad \partial_{ij} \Sigma = \bold{0}
\]
From this it follows that $\partial_{ij} (\Sigma^{-1}) = O\left(\frac{1}{q_i^2q_j^2}\right)$. The same holds for $\frac{\partial^2 f}{\partial q_i \partial q_j}$ because of (\ref{eq:der1}), completing the proof of the lemma.  
\end{proof}

\subsection{Dynamic Blackwell Comparison}\label{appxBlackwell}

\subsubsection{The Lemma}
This subsection establishes a dynamic version of Blackwell dominance for sequences of normal signals. As an overview, we first generalize \cite{Greenshtein} and show that a \emph{deterministic} (i.e. history-independent) signal sequence yields higher expected payoff than another in every intertemporal decision problem if (and only if) the former sequence induces lower posterior variances about $\theta_1$ at every period. This will be a corollary of the lemma below, which also covers strategies that may condition on signal realizations. 

We introduce some notation: Since $\theta_1$ is the only payoff-relevant state, the DM in our model only needs to remember the expected value of $\theta_1$ and the covariance matrix over all of the states (that is, expected values of the other states do not matter). Thus, we can summarize any history of beliefs by $h^T = (\mu^0_1, V^0; \dots, \mu^T_1, V^T)$, with $\mu^t_1$ representing the posterior expected value of $\theta_1$ after $t$ periods and $V^t$ the posterior covariance matrix. Since the posterior covariance matrix is a function of signal counts, we can also keep track of the evolution of posterior covariance matrices by a sequence of division vectors. That is, we will write the history as $h^T = (\mu^0_1, d(0); \dots; \mu^T_1, d(T))$, where each $d(t) = (d_1(t), \dots, d_K(t))$ counts the number of each signal acquired by time $t$. We can then view any information acquisition strategy $S$ as a mapping from such sequences of expected values and division vectors to signal choices. 

Consider a mapping $\tilde{G}$ from possible sequences of divisions to these sequences themselves: For each $(d(0), \dots, d(T))$, $\tilde{G}$ maps to another sequence $(\tilde{d}(0), \dots, \tilde{d}(T))$, subject to the following ``consistency" requirements. First, $\sum_{i}\tilde{d}_i(t) = t$, meaning that each $\tilde{d}(t)$ must be a possible division at time $t$. Second, $\tilde{d}_i(t) \geq \tilde{d}_i(t-1)$, meaning that the sequence $\tilde{d}$ can be attained via a sequential sampling rule. Lastly, we require
\[
(\tilde{d}(0), \dots, \tilde{d}(T-1)) = \tilde{G}(d(0), \dots, d(T-1))
\]
so that nesting sequences are mapped to nesting sequences. 

The following lemma says that if $d(\cdot)$ represents the division vectors under an information acquisition strategy $S$, and if $\tilde{G}$ is a consistent mapping that uniformly reduces the posterior variance, then we can find another information acquisition strategy $\tilde{S}$ whose division vectors are given by $\tilde{d}(\cdot)$. Moreover, our construction ensures that $\tilde{S}$ leads to more dispersed posterior beliefs than $S$ at every period, so that in any decision problem, acquiring signals according to $\tilde{S}$ is weakly better than $S$ (when actions are taken optimally). 

\begin{lemma}\label{lemmDynamicBlackwell}
Fix any information acquisition strategy $S$ and any consistent mapping $\tilde{G}$ defined above. Suppose that for every sequence of divisions $(d(0), \dots, d(T))$ realized under $S$, it holds that 
\[
f(\tilde{d}(T)) \leq f(d(T)).
\]
Then there exists deviation strategy $\tilde{S}$ such that, at every period $T$, any history $h^T = (\mu^0_1, d(0); \dots; \mu^T_1, d(T))$ under $S$ can be ``associated with'' a distribution of histories $\tilde{h}^T = (\nu^0_1, \tilde{d}(0); \dots; \nu^T_1, \tilde{d}(T))$ with the following properties:
\begin{enumerate}
\item the probability of $h^T$ occurring under $S$ is the same as the probability of its associated $\tilde{h}^T$ (integrated with respect to the probability of ``association") occurring under $\tilde{S}$; 

\item the total probability that any $\tilde{h}^T$ is associated to (integrated with respect to different possible $h^T$) is $1$; 

\item under the association, the distribution of $\nu^t_1$ is normal with mean $\mu^t_1$ and variance $f(d(t)) - f(\tilde{d}(t))$ for each $t$. 
\end{enumerate}
Consequently, for any decision strategy $A$, there exists another decision strategy $\tilde{A}$ such that the expected payoff under $(\tilde{S}, \tilde{A})$ is no less than the expected payoff under $(S, A)$. 
\end{lemma}

To interpret, the first two properties require that the association relation is a \emph{Markov kernel} between histories under $S$ and histories under $\tilde{S}$; this enables us to compare payoffs under $\tilde{S}$ to those under $S$. The third property guarantees that the alternative strategy $\tilde{S}$ is more informative than $S$. 

We note the following corollary, which is obtained from the previous lemma by considering a constant mapping $\tilde{G}$. 
\begin{corollary}\label{corDynamicBlackwell}
Define the $t$-optimal division vectors as in Section \ref{sec:intuition}. Suppose each coordinate of $n(Bt)$ increases in $t$. Then it is optimal for the DM to achieve $n(Bt)$ at every period. 
\end{corollary}

\subsubsection{Proof of Lemma \ref{lemmDynamicBlackwell}}
We construct $\tilde{S}$ iteratively as follows. In the first period, consider the signal choice under $S$ (given the null history). This signal leads to the division $d(1)$. Let $\tilde{S}$ observe the unique signal that would achieve the division $\tilde{d}(1)$.

After the first observation, the DM's \emph{distribution of posterior beliefs} about $\theta_1$ under $S$ is 
$\theta_1 \sim \mathcal{N}(\mu^1_1, f(d(1)))$ with $\mu^1_1$ a normal random variable with mean $\mu^0_1$ and variance $ f(\mathbf{0}) - f(d(1))$. By comparison, the distribution of posterior beliefs under $\tilde{S}$ is $\theta_1 \sim \mathcal{N}(\nu^1_1, f(\tilde{d}(1)))$ with $\nu^1_1$ drawn from $\mathcal{N}(\mu^0_1, f(\mathbf{0}) - f(\tilde{d}(1)))$. Since $f(\tilde{d}(1)) \leq f(d(1))$, the latter distribution of beliefs (under $\tilde{S}$) is more informative a la Blackwell. Thus, we can associate each belief $\theta_1 \sim \mathcal{N}(\mu^1_1, f(d(1)))$ under $S$ with a more informative distribution of beliefs $\mathcal{N}(\nu^1_1, f(\tilde{d}(1)))$ under $\tilde{S}$. To be more specific, for fixed $\mu^1_1$, the associated $\nu^1_1$ is distributed normally with mean  $\mu^1_1$ and variance $f(d(1)) - f(\tilde{d}(1)))$. Thus by construction, all three properties are satisfied at period $1$. To facilitate the discussion below, we say this distribution of beliefs under $\tilde{S}$ ``imitates" the belief $(\mu^1_1, f(d(1)))$ under $S$. 

In the second period, the deviation strategy $\tilde{S}$ takes the current belief $(\nu^1_1, f(\tilde{d}(1)))$ and randomly selects some $\mu^1_1$ (with conditional probabilities under the Markov kernel) to ``imitate." That is, given any selection of $\mu^1_1$, find the signal that $S$ would observe in period $2$ given belief $(\mu^1_1, f(d(1)))$. This signal choice under $S$ leads to the division sequence $(d(0), d(1), d(2))$, which is mapped to $(\tilde{d}(0), \tilde{d}(1), \tilde{d}(2))$. Naturally, we let $\tilde{S}$ observe the signal that would lead to the division $\tilde{d}(2)$. Such a signal is well-defined due to our consistency requirements on $\tilde{G}$. 

To proceed with the analysis, let us fix $\mu^1_1$ and study the distribution of posterior beliefs about $\theta_1$ after two observations. Under $S$, the distribution of posterior beliefs is $\theta_1 \sim \mathcal{N}(\mu^2_1, f(d(2)))$ with $\mu^2_1$ normally distributed with mean $\mu^1_1$ and variance $f(d(1)) - f(d(2))$. While under $\tilde{S}$, the distribution of posterior beliefs is $\theta_1 \sim (\nu^2_1, f(\tilde{d}(2)))$ with $\nu^2_1$ drawn from $\mathcal{N}(\mu^{1}_{1}, f(d(1)) - f(\tilde{d}(2)))$.\footnote{Here we use the following technical result: suppose the DM is endowed with \emph{a distribution of prior beliefs} $\theta \sim \mathcal{N}(\mu, V)$, with $\mu_1$ normally distributed with mean $y$ and variance $\sigma^2$, then upon observing signal $i$ and performing Bayesian updating, his distribution of posterior beliefs is $\theta \sim \mathcal{N}(\hat{\mu}, \hat{V})$, with $\hat{\mu}_1$ normally distributed with mean $y$ and variance $\sigma^2 + [V - \hat{V}]_{11}$. This is proved by noting that the DM's distribution of beliefs about $\theta_1$ must integrate to the same ex-ante distribution of $\theta_1$.} 

Since $f(\tilde{d}(2)) \leq f(d(2))$, the distribution of beliefs under $\tilde{S}$ Blackwell-dominates the distribution under $S$, for each $\mu^{1}_{1}$. We can thus associate each history $(\mu^1_1, d(1); \mu^2_1, d(2))$ under $S$ with a distribution of histories $(\nu^1_1, \tilde{d}(1); \nu^2_1, \tilde{d}(2))$ under $\tilde{S}$, such that the corresponding beliefs under $\tilde{S}$ are more informative at both periods. Repeating this procedure completes the construction of $\tilde{S}$, which satisfies all three properties stated in the lemma. 

Finally, suppose $A$ is any decision strategy that maps histories to actions. We need to find $\tilde{A}$ such that the pair $(\tilde{S}, \tilde{A})$ does no worse than $(S, A)$. This is straightforward given what we have done: at any history $\tilde{h}^{T}$ under $\tilde{S}$, let $\tilde{h}^T$ randomly select $h^T$ to imitate, and define $\tilde{A}(\tilde{h}^T) = A(h^T)$. Then we see that a DM who follows the decision strategy $A$ obtains the same payoff along any belief history $h$ as another DM who uses the decision strategy $\tilde{A}$ and faces the distribution of belief histories $\tilde{h}$. Integrating over $h$, we have shown that $(\tilde{S}, \tilde{A})$ achieves the same payoff as $(S, A)$. The lemma is proved.

\subsection{Proof of Theorem \ref{propBatch} (Large Block of Signals)} \label{appxBatch}
By Corollary \ref{corDynamicBlackwell}, it suffices to show that for sufficiently large $B$, each coordinate $n(Bt)$ is increasing in $t$. To do this, we first argue that the signal counts grow to infinity (roughly) proportionally. In more detail, define 
\begin{equation}\label{eq:lambda}
\lambda_i = \frac{\lvert [C^{-1}]_{1i} \rvert \cdot \sigma_i}{ \sum_{j=1}^{K} \lvert [C^{-1}]_{1j} \rvert \cdot \sigma_j }.
\end{equation}
Then we will show that for each signal $i$, $n_i(t) - \lambda_i \cdot t$ remains bounded even as $t \to \infty$.

Indeed, we must at least have $n_i(t) \to \infty$; otherwise the posterior variance $f(n(t))$ would be bounded away from zero, which would contradict the optimality of $n(t)$ since $f(t/K, \dots, t/K) \to 0$. Additionally, we compute from (\ref{eq:f}) that
\begin{equation}\label{eq:firstDer}
\partial_i f(n(t)) = -\frac{\sigma_i^2}{n_i^2} \cdot [V^0C' \Sigma^{-1} \Delta_{ii} \Sigma^{-1} CV^0]_{11}. 
\end{equation}

As each $n_i \to \infty$, the matrix $\Sigma = CV^0C' + D^{-1}$ (see Lemma \ref{lemmVar}) converges to $CV^0C'$. So $V^0C' \Sigma^{-1} \Delta_{ii} \Sigma^{-1} CV^0$ converges to the matrix $Q_i$ defined in (\ref{eq:Q}). It follows from (\ref{eq:firstDer}) that $\partial_i f \sim \frac{-\sigma_i^2}{n_i^2} \cdot [Q_i]_{11}$ (ratio converges to $1$). Since a $t$-optimal division must satisfy $\partial_i f \sim \partial_j f$ (because we are doing discrete optimization, $\partial_i f$ and $\partial_j f$ need only be approximately equal), we deduce that $n_i$ and $n_j$ must grow proportionally. Using $[Q_i]_{11} = ([C^{-1}]_{1i})^2$, we have $n_i(t) \sim \lambda_i t$. 

Next, note that because $n_i(t) \sim \lambda_i t$, $\Sigma = CV^0C' + D^{-1} = CV^0C' + O(\frac{1}{t})$.\footnote{``Big O" notation has the usual meaning.} Thus in fact $V^0C' \Sigma^{-1} \Delta_{ii} \Sigma^{-1} CV^0$ converges to $Q_i$ at the rate of $\frac{1}{t}$. From (\ref{eq:firstDer}), we obtain $\partial_i f = \frac{-\sigma_{i}^2 \cdot [Q_i]_{11} + O(\frac{1}{t})}{n_i^2}$. $t$-optimality gives us the first-order condition $\partial_i f = \partial_j f + O(\frac{1}{t^3})$.\footnote{That is, error terms due to discreteness are no larger than $\frac{1}{t^3}$. We omit the details.} So 
\[
\frac{\lambda_i^2 + O(\frac{1}{t})}{n_i^2} = \frac{\lambda_j^2 + O(\frac{1}{t})}{n_j^2}.
\]
This is equivalent to $\lambda_i^2 n_j^2 - \lambda_j^2 n_i^2 = O(t)$, which yields $\lambda_i n_j - \lambda_j n_i = O(1)$ after factorization. Hence $n_i(t) = \lambda_i \cdot t + O(1)$ as we claimed. 

Having completed this asymptotic characterization of the $t$-optimal division vectors, we will now show that $n(t+K-1) \geq n(t)$ (in each coordinate) whenever $t$ is sufficiently large. Theorem \ref{propBatch} will follow once this is proved.\footnote{To be fully rigorous, this only proves Theorem \ref{propBatch} when $B$ is sufficiently large and is a multiple of $K-1$. However, we can similarly show $n(t+K) \geq n(t)$ for sufficiently large $t$. The two inequalities $n(t+K-1) \geq n(t)$ and $n(t+K) \geq n(t)$ together are sufficient to deduce Theorem \ref{propBatch} for all large $B$.}

Suppose for the sake of contradiction that $n_1(t+K-1) \leq n_1(t) - 1$. Note we have $\sum_{i = 1}^{K} (n_i(t+K-1) - n_i(t)) = K-1$. So $\sum_{i = 2}^{K} (n_i(t+K-1) - n_i(t)) \geq K$, and we can without loss of generality assume $n_2(t+K-1) \geq n_2(t) + 2$. To summarize, when transitioning from $t$-optimality to $t+K-1$-optimality, signal $1$ is acquired at least once less and signal $2$ at least twice more. Below we will obtain a contradiction by arguing that at period $t + K - 1$, the posterior variance could be further reduced by observing signal $1$ once more and signal $2$ once less. 

Indeed, let us write $n_i = n_i(t)$ and $\tilde{n}_i = n_i(t+K-1)$. Then $t$-optimality of $n(t)$ gives us 
\[
f(n_1 - 1, n_2 + 1, \dots, n_K) \geq f(n_1, n_2, \dots, n_K).
\]
With a slight abuse of notation, we let $\partial_i f$ to denote the \emph{discrete} partial derivative of $f$: $\partial_i f (q) = f(q_i+1, q_{-i}) - f(q)$. Then the above display is equivalent to 
\begin{equation}\label{eq:ineq1}
\partial_2 f(n_1-1, n_2, \dots, n_K) \geq \partial_1 f(n_1-1, n_2, \dots, n_K).
\end{equation}
We claim this implies the following:
\begin{equation}\label{eq:ineq2}
\partial_2 f(\tilde{n}_1, \tilde{n}_2-1, \dots, \tilde{n}_K) > \partial_1 f(\tilde{n}_1, \tilde{n}_2-1, \dots, \tilde{n}_K).
\end{equation}
This would lead to 
\[
f(\tilde{n}_1, \tilde{n}_2, \dots, \tilde{n}_K) > f(\tilde{n}_1+1, \tilde{n}_2-1, \dots, \tilde{n}_K),
\]
which would be our desired contradiction. 

It remains to show (\ref{eq:ineq1})$\implies$(\ref{eq:ineq2}). By assumption, we have $\tilde{n}_1 \leq n_1 - 1$, $\tilde{n}_2 \geq n_2 + 2$ and the difference between any $\tilde{n}_j$ and $n_j$ is bounded uniformly over $t$. Thus the LHS of (\ref{eq:ineq2}) exceeds the LHS of (\ref{eq:ineq1}) by (at least) a second derivative $\partial_{22}$ minus a finite number of cross partial derivatives $\partial_{2j}$. By Lemma \ref{lemmSecondDer}, this difference on the LHS is positive with order $\frac{1}{t^3}$. The difference between the RHS of (\ref{eq:ineq2}) and the RHS of (\ref{eq:ineq1}) can be positive or negative, but either way it has order $O(\frac{1}{t^4})$. This shows (\ref{eq:ineq2}) is a consequence of (\ref{eq:ineq1}), and the theorem follows.

\subsection{Proof of Theorem \ref{propSeparable} (Separable Environments)} \label{appxSeparable}

Suppose the informational environment is separable. We will show $n(t)$ increases in $t$, which implies the theorem via Corollary \ref{corDynamicBlackwell}. 

Note that in a separable environment, the definition of $t$-optimality reduces to:
\[
n(t) =(n_1(t), \dots, n_K(t))\in \argmin_{(q_1,\dots,q_K) : q_i \in \mathbb{Z}_+, \sum_{i=1}^{K} q_i=t} \sum_{i = 1}^{K} g_i(q_i) 
\]
where $g_1, \dots, g_K$ are convex functions. 

In this setting, the myopic information acquisition strategy sequentially chooses the signal $i$ that minimizes the difference $g_i(q_i+1) - g_i(q_i)$, given the current division vector $q$. But since the $g$-functions are convex, the outcome under the myopic strategy coincides with $t$-optimality at every period $t$ (This is fairly well known, and it can be proved quickly by induction.) Hence $n(t)$ can and will be sequentially attained.

\subsection{Preparation for the Proof of Theorem \ref{propGenericEventual}} \label{appxDynamic}

\subsubsection{Switch Deviations}
We now introduce preliminary results that will be used to show that the optimal rule eventually proceeds myopically in generic environments. Relative to the proofs of Theorems \ref{propBatch} and \ref{propSeparable}, the new difficulty that arises is that in general, the optimal information acquisition strategy conditions on signal realizations. As a result, the induced division vectors $d(\cdot)$ are stochastic, and we will need the full power of our dynamic Blackwell lemma.

In what follows, we will apply Lemma \ref{lemmDynamicBlackwell} using a particular class of mappings $\tilde{G}$.

\begin{definition}
Fix a particular sequence of divisions $(d^*(0), d^*(1), \dots, d^*(t_0))$. Let $i$ be the signal observed in period $t_0$ and $j$ be any other signal. An \emph{$(i,j)$-switch} mapping $\tilde{G}$ specifies the following: 
\begin{enumerate}
\item Suppose $T < t_0$ or $d(t) \neq d^*(t)$ for some $t \leq t_0$, then let $\tilde{G}(d(0), \dots, d(T))$ be itself. 

\item Otherwise $T \geq t_0$ and $d(t) = d^*(t), \forall t \leq t_0$. If $d_j(T) = d_j(t_0)$, then let $\tilde{d}(T) = (d_i(T) - 1, d_j(T) + 1, d_{-ij}(T))$. If $d_j(T) > d_j(t_0)$, then let $\tilde{d}(T) = d(T)$. 
\end{enumerate}
\end{definition}

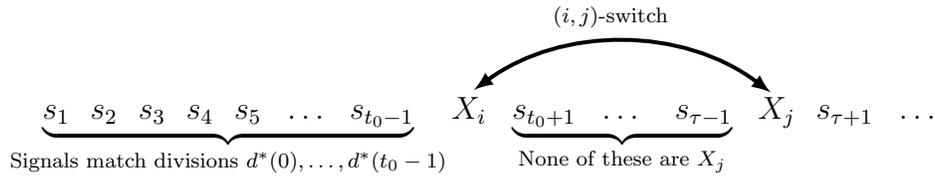
\begin{figure}[h]
\centering
\input{switch-fig.tex}
\caption{Pictorial representation of an $(i,j)$-switch based on a sequence of divisions $d^*(0),\ldots,d^*(t_0)$.}
\end{figure}

Let us interpret this definition by relating to the resulting deviation strategy $\tilde{S}$ constructed in Lemma \ref{lemmDynamicBlackwell}. The first case above says that $\tilde{S}$ only deviates when the history of divisions is $d^*(0), \dots, d^*(t_0-1)$ and $S$ is about to observe signal $i$ in period $t_0$. The second case says that $\tilde{S}$ dictates observing signal $j$ instead at that history; subsequently, $\tilde{S}$ observes the same signal as $S$ (at the imitated belief) until the first period at which $S$ is about to observe signal $j$. If that period exists, the deviation strategy $\tilde{S}$ \emph{switches back} to observing signal $i$ and coincides with $S$ afterwards. 

The benefit of these ``switch deviations" is that their posterior variances can be easily compared to the original strategy. Specifically, $\tilde{d}(t) = d(t)$ except at those histories that begin with $d^*(0), d^*(1), \dots, d^*(t_0-1)$ (and before signal $j$ is observed again under $S$). At such histories, the posterior variance is strictly lower under $\tilde{S}$ if and only if 
\[
f(d_i(t) - 1, d_j(t) + 1, d_{-ij}(t)) < f(d(t)). 
\]
Using (absolute values of) the discrete partial derivatives, we can rewrite this condition as
\begin{equation}\label{eq:comparePartials}\tag{*}
\lvert \partial_i f(d_i(t)-1, d_j(t), d_{-ij}(t)) \rvert < \lvert \partial_j f(d_i(t)-1, d_j(t), d_{-ij}(t)) \rvert.
\end{equation}

We can thus obtain the following corollary:
\begin{corollary}\label{corSwitch}
Suppose we can find a history of divisions $d(0), \dots, d(t_0)$ realized under $S$ such that $d_i(t_0) = d_i(t_0-1) + 1$ and moreover (\ref{eq:comparePartials}) holds for all divisions $d(t)$ with $d_j(t) = d_j(t_0)$ and $d_k(t) \geq d_k(t_0), \forall k$. Then the switch deviation $\tilde{S}$ constructed above improves upon $S$.
\end{corollary}
Note that the condition $d_j(t) = d_j(t_0)$ captures the fact that $\tilde{d}(t)$ differs from $d(t)$ only until signal $j$ is chosen again by $S$. Meanwhile, $d_k(t) \geq d_k(t_0), \forall k$ holds because we only compare posterior variances after $t_0$ periods.

\subsubsection{Asymptotic Characterization of Optimal Strategy}
Below we will use the contrapositive of Corollary \ref{corSwitch} to argue that if $S$ is the optimal information acquisition strategy, then we cannot find any history of realized divisions such that (\ref{eq:comparePartials}) always holds. Technically speaking, we might worry that although $\tilde{S}$ strictly improves upon $S$ in terms of posterior variances, it might achieve the same expected payoff as $S$ (for instance, when the DM faces a constant payoff function). Nonetheless, by Zorn's lemma we can choose $S$ to be an optimal strategy that is additionally ``un-dominated" in terms of posterior variances. With that choice, the deviation $\tilde{S}$ cannot exist, and our arguments remain valid. (Note that our theorems only state that the DM \emph{has an optimal strategy} \dots)

To illustrate, we now derive the asymptotic signal proportions for the optimal information acquisition strategy $S$. 
\begin{lemma}\label{lemmDynamicFreq}
Suppose $S$ is the optimal information acquisition strategy, and $d(\cdot)$ is its induced divisions. Let $\lambda_k$ be defined as in (\ref{eq:lambda}). In generic informational environments, the difference $d_k(T) - \lambda_k \cdot T$ remains bounded as $T \to \infty$, for any realized division $d(T)$ and each signal $k$. 
\end{lemma}

\begin{proof}
For this proof, we only need the informational environment to be such that each signal has \emph{strictly positive marginal value}. That is, for any signal $k$ and any possible division $q$, we require
\[
f(q_k + 1, q_{-k}) < f(q).
\]
This is ``generically" satisfied because any equality $f(q_k + 1, q_{-k}) = f(q)$ would impose a non-trivial polynomial equation over the signal linear coefficients, and the number of such constraints is at most countable. 

Under this genericity assumption, let us first show $d_k(T) \to \infty$ holds for each signal $k$, and the speed of divergence depends only on the informational environment. For contradiction, suppose this is not true. Then we can find a sequence of histories $\{h^{T_m}\}$ such that $T_m \to \infty$ but $d_1(T_m)$ remains bounded (these histories need not nest one another). By passing to a subsequence, we may assume $q_k = \lim_{m \to \infty} d_k(T_m)$ exists for every signal $k$, where this limit may be infinity. Define $I$ to be the non-empty subset of signals (not including signal $1$) with $q_k = \infty$. Furthermore, we assume that the signal observed in the last period of each of these histories $h^{T_m}$ is the same signal $i$. We also assume $i \in I$; otherwise just truncate the histories by finitely many periods.  

Take any signal $j \notin I$ (for instance, $j = 1$ works). Choose $T_m$ sufficiently large and consider the $(i,j)$-switch deviation $\tilde{S}$ that deviates from $h^{T_m}$ by observing signal $j$ instead of $i$ in period $T_m$. 
We will verify (\ref{eq:comparePartials}) for all possible divisions $d(t)$ with $d_j(t) = d_j(T_m)$ and $d_i(t) \geq d_i(T_m)$, which will contradict the optimality of $S$ via Corollary \ref{corSwitch}. Indeed, note that as $T_m \to \infty$, $d_i(T_m) \to \infty$ because $i \in I$. Since $d_i(t) \geq d_i(T_m)$, the LHS of (\ref{eq:comparePartials}) approaches zero as $T_m$ increases. By comparison, the RHS of (\ref{eq:comparePartials}) is bounded away from zero because $d_j(t) = d_j(T_m)$ is bounded, and we assume each signal has strictly positive marginal value. Hence (\ref{eq:comparePartials}) holds and we have shown that $d(T) \to \infty$ in each coordinate. 

Next, from (\ref{eq:firstDer}), we have the following approximations for the partial derivatives:
\[
\lvert \partial_i f(d_i(t)-1, d_j(t), d_{-ij}(t)) \rvert \sim \frac{\sigma_i^2 \cdot [Q_i]_{11}}{d_i(t)^2} \qquad \lvert \partial_j f(d_i(t)-1, d_j(t), d_{-ij}(t)) \rvert \sim \frac{\sigma_j^2 \cdot [Q_j]_{11}}{d_j(t)^2}.
\]
If $\limsup_{t_0 \to \infty} \frac{d_i(t_0)}{d_j(t_0)} > \frac{\lambda_i}{\lambda_j}$ (recall that $\lambda_i$ is proportional to $\sigma_i \cdot \sqrt{[Q_i]_{11}}$), then the above estimates would imply (\ref{eq:comparePartials}) whenever $d_i(t) \geq d_i(t_0)$ (because $t \geq t_0$) and $d_j(t) = d_j(t_0)$. That would contradict the optimality of $S$. Hence, $\limsup_{t_0 \to \infty} \frac{d_i(t_0)}{d_j(t_0)} \leq \frac{\lambda_i}{\lambda_j}$ for every pair of signals $i$ and $j$. It follows that $d_k(t_0) \sim \lambda_k \cdot t_0, \forall k$. 

Once these asymptotic proportions are proved, we know that the matrix $\Sigma = CV^0C' + D^{-1}$ converges to $CV^0C'$ at the rate of $\frac{1}{t}$. By (\ref{eq:firstDer}), we can deduce more precise approximations:
\[
\lvert \partial_i f(d_i(t)-1, \cdots) \rvert = \frac{\sigma_i^2 \cdot [Q_i]_{11} + O(\frac{1}{t})}{d_i(t)^2} \qquad \lvert \partial_j f(d_i(t)-1, \cdots) \rvert = \frac{\sigma_j^2 \cdot [Q_j]_{11} + O(\frac{1}{t})}{d_j(t)^2}.
\]
If $\frac{d_i(t_0)}{d_j(t_0)} > \frac{\lambda_i}{\lambda_j} + O(\frac{1}{t_0})$, then these refined estimates would again imply (\ref{eq:comparePartials}) whenever $d_i(t) \geq d_i(t_0)$ and $d_j(t) = d_j(t_0)$. To avoid the resulting contradiction, we must have $\frac{d_i(t_0)}{d_j(t_0)} \leq \frac{\lambda_i}{\lambda_j} + O(\frac{1}{t_0})$
for every signal pair. This enables us to conclude $d_k(t_0) = \lambda_k \cdot t_0 + O(1)$ as desired. 
\end{proof}

\subsection{Proof of Theorem \ref{propGenericEventual} (Generic Eventual Myopia)}\label{appxGenericEventual}

\subsubsection{Outline of the Proof}
To guide the reader through this appendix, we begin by outlining the proof of the theorem, which is broken down into several steps. Throughout, we focus on the case of $B = 1$ (one signal each period), but our proof easily extends to arbitrary $B$. We will first show a simpler (and weaker) result that, in generic environments, the number of periods in which the optimal strategy differs from the $t$-optimal division has natural density $1$. Our proof of this result is based on the observation that if equivalence does not hold at some time $t$, there must be \emph{two different divisions} over signals for which the resulting posterior variances about $\theta_1$ are within $O(\frac{1}{t^4})$ from each other. This leads to a Diophantine approximation inequality, which we can show only occurs at a vanishing fraction of periods $t$. 

To improve the result and demonstrate equivalence at \emph{all late periods}, we show that the number of ``exceptional periods" $t$ is generically finite if there are \emph{three different divisions} over signals whose posterior variances are within $O(\frac{1}{t^4})$ from each other. This allows us to conclude that in generic environments, the $t$-optimal divisions eventually monotonically increase in $t$. 

In such environments, $t$-optimality can be achieved at every late period. Thus, whenever $t$-optimality obtains in \emph{some} late period, it will be sustained in all future periods. Since we have already established that the optimal strategy achieves $t$-optimality infinitely often, we conclude equivalence at all large $t$. 

We highlight that in this appendix, we use a slightly different notion of ``generic" where we fix the signal coefficient matrix $C$ and instead (randomly) vary the signal variances $\{\sigma_i^2\}$. This concept implies (and is stronger than) the previous genericity concept defined on $C$. 

\subsubsection{Equivalence at Almost All Times}
We begin by proving a weaker result, that the optimal strategy induces the $t$-optimal division $n(t)$ at almost all periods $t$. 

\begin{proposition}\label{propGenericAlmostAll}
Suppose the informational environment $(V^0, C, \{\sigma_{i}^2\})$ has the property that for any $i \neq j$, the ratio $\frac{\lambda_i}{\lambda_j}$ is an irrational number. Then, at a set of times with natural density $1$,\footnote{Formally, for any set of positive integers $A$, let $A(N)$ count the number of integers in $A$ no greater than $N$. Then we define the natural density of $A$ to be $\lim_{N \to \infty} \frac{A(N)}{N}$, when this limit exists.} $d(t) = n(t)$ (which is unique) holds for every decision problem. In particular, the optimal strategy induces a deterministic division vector at such times. 
\end{proposition}

\begin{proof}[Proof of Proposition \ref{propGenericAlmostAll}]
Suppose that $d_1(T) \geq n_1(T)+1$ and $d_2(T) \leq n_2(T)-1$. Consider the last period $t_0 \leq T$ in which the optimal strategy observed signal $1$. Then 
\[
d_1(t_0) = d_1(T) \geq n_1(T) + 1; \quad d_2(t_0) \leq d_2(T) \leq n_2(T) - 1.
\]

Using the contrapositive of Corollary \ref{corSwitch} with the $(1,2)$-switch, we know that (\ref{eq:comparePartials}) cannot always hold. Thus there exists a division $d(t)$ such that the inequality (\ref{eq:comparePartials}) is reversed. That is, we can find a division $d(t)$ with $d_1(t) \geq d_1(t_0)$ and $d_2(t) = d_2(t_0)$ such that (getting rid of the absolute values)
\begin{equation}\label{eq:ineq3}
\partial_1 f(d_1(t)-1, d_2(t), d_{-12}(t)) \leq \partial_2 f(d_1(t)-1, d_2(t), d_{-12}(t)). 
\end{equation}
On the other hand, $t$-optimality of $n(t)$ gives us
\begin{equation}\label{eq:ineq4}
\partial_1 f(n_1(T), n_2(T)-1, n_{-12}(T)) \geq \partial_2 f(n_1(T), n_2(T)-1, n_{-12}(T)). 
\end{equation}

Note that $d_2(t) = d_2(t_0)$ implies $t - t_0$ is bounded (due to Lemma \ref{lemmDynamicFreq}). On the other hand, we have $d_1(t) = d_1(T_0)$ by construction ($t_0$ is the last period signal $1$ was observed). Hence $t_0 - T$ is also bounded. Combining both, we deduce $t - T$ must be bounded. Applying Lemma \ref{lemmDynamicFreq} again, we know that any difference $d_i(t) - n_i(T)$ is bounded. 

Now because $d_1(t) - 1 \geq d_1(t_0) - 1 \geq n_1(T)$, the LHS of (\ref{eq:ineq3}) has size at least the LHS of (\ref{eq:ineq4}) minus a finite number of cross partial derivatives $\partial_{1j}$. Similarly, the RHS of (\ref{eq:ineq3}) is at most bigger than the RHS of (\ref{eq:ineq3}) by a number of cross partials. Together with the order difference lemma, these imply that the only way (\ref{eq:ineq3}) and (\ref{eq:ineq4}) can both hold is if the two sides of (\ref{eq:ineq4}) differ by at most $O\left(\frac{1}{T^4}\right)$. 

To summarize: A necessary condition for $d_1(T) \geq n_1(T)+1$ and $d_2(T) \leq n_2(T) + 1$ to occur is that 
\begin{equation}\label{eq:tinyDiff1}
|f(n_1(T)+1, n_2(T)-1, \dots, n_K(T)) - f(n(T))| = O\left(\frac{1}{T^4}\right). 
\end{equation}
Hence, to prove the Proposition we only need to show that (\ref{eq:tinyDiff1}) holds at a set of times with natural density $0$. The following lemma proves exactly this property.
\end{proof}

\smallskip

\begin{lemma}\label{lemmTwoDivisions}
Suppose $\frac{\lambda_1}{\lambda_2}$ is an irrational number. For positive constants $c_0, c_1$, define $\mathcal{A}(c_0, c_1)$ to be the following set of positive integers: 
\begin{align*}
\{t: \exists \,\, q_1, q_2, \dots, q_K & \in \mathbb{Z}_{+}, s.t. ~|q_i - \lambda_i \cdot t| \leq c_0, \forall i \\
& ~ \wedge ~ |f(q_1, q_2+1, \dots, q_K) - f(q_1+1, q_2, \dots, q_K)| \leq c_1/t^4 \}. 
\end{align*}
Then $\mathcal{A}(c_0, c_1)$ has natural density zero.
\end{lemma}

\begin{proof}[Proof of Lemma \ref{lemmTwoDivisions}]
The proof relies on the following technical result, which gives a precise approximation of the \emph{discrete partial derivatives} of $f$:
\begin{lemma}\label{lemmFirstDerApprox}
Fix the informational environment. There exists a constant $a_j$ such that 
\begin{equation}\label{eq:firstDerApprox}
f(q_j, q_{-j}) - f(q_j+1, q_{-j}) = \frac{\sigma_j^2 \cdot [Q_j]_{11}}{(q_j-a_j)^2} + O\left( \frac{1}{t^4}; c_0 \right)
\end{equation}
holds for all $q_1, \dots, q_K$ with $|q_i - \lambda_i t| \leq c_0, \forall i$. The notation $O\left( \frac{1}{t^4}; c_0 \right)$ means an upper bound of $\frac{L}{t^4}$, where the constant $L$ may depend on the informational environment as well as on $c_0$.\footnote{In applying Lemma \ref{lemmTwoDivisions} to prove Proposition \ref{propGenericAlmostAll}, $c_0$ is taken to be the bound on $n_i - \lambda_i \cdot t$.}
\end{lemma}

Assuming (\ref{eq:firstDerApprox}), we see that the condition 
\[
\left|f(q_1, q_2+1, \dots, q_K) - f(q_1+1, q_2, \dots, q_K)\right| \leq \frac{c_1}{t^4}
\]
implies $\left|\frac{\sigma_{1}^2 \cdot [Q_1]_{11}}{(q_1-a_1)^2} - \frac{\sigma_{2}^2 \cdot [Q_2]_{11}} {(q_2 - a_2)^2}\right| \leq \frac{c_2}{t^4}$ and thus $\left|(\frac{\lambda_1}{q_1-a_1})^2 - (\frac{\lambda_2}{q_2-a_2})^2\right| \leq \frac{c_3}{t^4}$ for some larger positive constants $c_2, c_3$. This further implies $\left|\frac{\lambda_1}{q_1-a_1} - \frac{\lambda_2}{q_2-a_2}\right| \leq \frac{c_4}{t^3}$, which reduces to 
\begin{equation}\label{eq:Dio1}
\left|q_2 - a_2 - \frac{\lambda_2}{\lambda_1}(q_1 - a_1)\right| \leq \frac{c_5}{t}. 
\end{equation}
This inequality says that the fractional part of $\frac{\lambda_2}{\lambda_1} q_1$ is very close to the fractional part of $\frac{\lambda_2}{\lambda_1} a_1 - a_2$. But since $\frac{\lambda_2}{\lambda_1}$ is an irrational number, the fractional part of $\frac{\lambda_2}{\lambda_1} q_1$ is ``equi-distributed'' in (0,1) as $q_1$ ranges in the positive integers.\footnote{The Equi-distribution Theorem states that for any irrational number $\alpha$ and any sub-interval $(a,b) \subset (0,1)$, the set of positive integers $n$ such that the fractional part of $\alpha n$ belongs to $(a,b)$ has natural density $b-a$. It is a special case of the Ergodic Theorem.} Thus the Diophantine approximation (\ref{eq:Dio1}) only has solution at a set of times $t$ with natural density $0$, proving Lemma \ref{lemmTwoDivisions}. Below we supply the technically involved proof of (\ref{eq:firstDerApprox}). 
\end{proof}

\smallskip

\begin{proof}[Proof of Lemma \ref{lemmFirstDerApprox}]
Fix $q_1, \dots, q_K$ and the signal $j$. Recall the diagonal matrix $D = \diag(\frac{q_1}{\sigma_1^2}, \dots, \frac{q_K}{\sigma_K^2})$. Consider any $\hat{q}_j \in [q_j, q_j + 1]$ and let $\hat{D}$ be the analogue of $D$ for the division $(\hat{q}_j, q_{-j})$. That is, $\hat{D} = D$ except that $[\hat{D}]_{jj} = \frac{\hat{q}_j}{\sigma_j^2}$. Let $\hat{\Sigma} = CV^0C' + \hat{D}^{-1}$. From (\ref{eq:firstDer}), we have 
\begin{equation}\label{eq:firstDerApprox1}
\partial_j f(\hat{q}_j, q_{-j}) = -\frac{\sigma_{j}^2}{\hat{q}_j^2} \cdot \left[V^0C' \hat{\Sigma}^{-1} \Delta_{jj} \hat{\Sigma}^{-1} CV^0\right]_{11}. 
\end{equation}
\emph{Here and later in this proof, $\partial_j f$ represents the usual continuous derivative rather than the discrete derivative.}

Let $D_0 = \diag\left(\frac{\lambda_1 t}{\sigma_1^2}, \dots, \frac{\lambda_K t}{\sigma_K^2}\right)$ and $\Sigma_0 = CV^0C' + D_0^{-1}$. For $|q_i - \lambda_i t| \leq c_0, \forall i$ we have $\hat{D} - D_0 = O(c_0)$, where the Big O notation applies entry-wise. It follows that 
\[
\hat{\Sigma} = CV^0C' + \hat{D}^{-1} = CV^0C' + D^{-1} + O(\frac{1}{t^2}; c_0) = \Sigma_0 + O(\frac{1}{t^2}; c_0). 
\]
Observe that the matrix inverse is a differentiable mapping at $\Sigma_0$ (which is $CV^0C' + D_0^{-1} \succeq CV^0C'$ and thus positive definite). Thus we have 
\[
\hat{\Sigma}^{-1} = \Sigma_0^{-1} +  O\left(\frac{1}{t^2}; c_0\right).
\]
Plugging this into (\ref{eq:firstDerApprox1}) and using $\hat{q}_j \sim \lambda_j t$, we obtain that 
\begin{equation}\label{eq:firstDerApprox2}
\partial_j f(\hat{q}_j, q_{-j}) = -\frac{\sigma_{j}^2}{\hat{q}_j^2} \cdot \left[V^0C' \Sigma_0^{-1} \Delta_{jj} \Sigma_0^{-1} CV^0\right]_{11} + O\left(\frac{1}{t^4}; c_0\right).
\end{equation}

Since $\Sigma_0 = CV^0C' + \frac{1}{t} \cdot \diag\left(\frac{\sigma_{1}^2}{\lambda_1}, \dots, \frac{\sigma_{K}^2}{\lambda_K}\right)$, we can apply Taylor expansion (to the matrix inverse map) and write
\begin{equation}\label{eq:Sigma0Inverse}
\Sigma_0^{-1} = (CV^0C')^{-1} - \frac{1}{t}(CV^0C')^{-1} \cdot \diag\left(\frac{\sigma_{1}^2}{\lambda_1}, \dots, \frac{\sigma_{K}^2}{\lambda_K}\right) \cdot (CV^0C')^{-1} + O\left(\frac{1}{t^2}\right).
\end{equation}
This implies
\begin{align}\label{eq:Taylor}
V^0C'\Sigma_0^{-1} \Delta_{jj} \Sigma_0^{-1} CV^0 &= V^0C' (CV^0C')^{-1} \Delta_{jj} (CV^0C')^{-1} CV^0 - \frac{M_j}{t} + O\left(\frac{1}{t^2}\right)\nonumber \\ 
&= Q_j - \frac{M_j}{t} + O\left(\frac{1}{t^2}\right),
\end{align}
where $M_j$ is a fixed $K \times K$ matrix depending only on the informational environment. For future use, we note that 
\begin{align}\label{eq:Mj}
M_j &= V^0C' (CV^0C')^{-1} \diag\left(\frac{\sigma_{1}^2}{\lambda_1},\dots, \frac{\sigma_{K}^2}{\lambda_K}\right) (CV^0C')^{-1} \Delta_{jj} (CV^0C')^{-1}CV^0 \nonumber \\
&\quad \quad \quad + V^0C' (CV^0C')^{-1} \Delta_{jj} (CV^0C')^{-1} \diag\left(\frac{\sigma_{1}^2}{\lambda_1},\dots, \frac{\sigma_{K}^2}{\lambda_K}\right) (CV^0C')^{-1} CV^0 \nonumber \\
&= C^{-1} \diag\left(\frac{\sigma_{1}^2}{\lambda_1},\dots, \frac{\sigma_{K}^2}{\lambda_K}\right) (CV^0C')^{-1} \Delta_{jj} C'^{-1} \nonumber \\
&\quad \quad \quad + C^{-1} \Delta_{jj} (CV^0C')^{-1}\diag\left(\frac{\sigma_{1}^2}{\lambda_1},\dots, \frac{\sigma_{K}^2}{\lambda_K}\right)C'^{-1}.
\end{align}

Using (\ref{eq:Taylor}), we can simplify (\ref{eq:firstDerApprox2}) to 
\begin{equation}\label{eq:firstDerApprox3}
\partial_j f(\hat{q}_j, q_{-j}) =- \frac{\sigma_{j}^2}{\hat{q}_j^2}\cdot  \left[Q_j - \frac{M_j}{t}\right]_{11} + O\left(\frac{1}{t^4};c_0\right). 
\end{equation}
Integrating this over $\hat{q}_j \in [q_j, q_j+1]$, we conclude that
\begin{equation}\label{eq:firstDerApprox4}
f(q_j, q_{-j}) - f(q_j+1, q_{-j}) = \frac{\sigma_{j}^2}{q_j(q_j+1)}\cdot  \left[Q_j - \frac{M_j}{t}\right]_{11} + O\left(\frac{1}{t^4};c_0\right).
\end{equation}
We set $a_j = -\left(\frac{\lambda_j \cdot [M_j]_{11}}{2[Q_j]_{11}} + \frac{1}{2}\right)$. Then
\[
\frac{\sigma_{j}^2}{q_j(q_j+1)}\cdot  \left[Q_j - \frac{M_j}{t}\right]_{11} = (\sigma_j^2 \cdot [Q_j]_{11}) \cdot \frac{1 + \frac{2a_j + 1}{\lambda_j t}}{q_j (q_j+1)} = \frac{\sigma_j^2 \cdot [Q_j]_{11}}{(q_j-a_j)^2} + O\left(\frac{1}{t^4}; c_0\right), 
\]
implying the desired approximation (\ref{eq:firstDerApprox}). The last equality above uses 
$\frac{1 + \frac{2a_j + 1}{\lambda_j t}}{q_j (q_j+1)} = \frac{1}{(q_j-a_j)^2} + O\left(\frac{1}{t^4}; c_0\right)$, which is because
\[
\frac{q_j (q_j+1)}{(q_j-a_j)^2} = 1 + \frac{2(a_j+1)}{q_j - a_j} + O\left(\frac{1}{(q_j-a_j)^2}\right) = 1 + \frac{2a_j + 1}{\lambda_j t} + O\left(\frac{1}{t^2}; c_0\right)
\]
dividing through by $q_j(q_j + 1)$. 
\end{proof}

\subsubsection{A Simultaneous Diophantine Approximation Problem}
The above Lemma \ref{lemmTwoDivisions} tells us that at most times $t$, there do not exist a \emph{pair} of divisions (differing minimally on two signal counts) that lead to posterior variances close to each other (with a difference of $\frac{c_1}{t^4}$). We obtain a stronger result if a \emph{triple} of such divisions were to exist.

\begin{lemma}\label{lemmThreeDivisions}
Fix $V^0$ and $C$, and let signal variances vary. For positive constants $c_0, c_1$, define $\mathcal{A}^*(c_0, c_1)$ to be the following set of positive integers:
\begin{align*}
\{t: \exists \,\, q_1, q_2, q_3, \dots, q_K & \in \mathbb{Z}_{+}, s.t. ~ |q_i - \lambda_i t| \leq c_0, \forall i \\
&  ~ \wedge ~ |f(q_1, q_2+1, q_3, \dots, q_K) - f(q_1+1, q_2, q_3, \dots, q_K)| \leq c_1/t^4 \\
& ~ \wedge ~ |f(q_1, q_2, q_3+1, \dots, q_K) - f(q_1+1, q_2, q_3, \dots, q_K)| \leq c_1/t^4 \}
\end{align*}
Then, for generic signal variances, $\mathcal{A}^*(c_0, c_1)$ has finite cardinality.
\end{lemma}

\begin{proof}
So far we have been dealing with fixed informational environments. However, a number of parameters defined above depend on the signal variances $\sigma = \{\sigma_i^2\}_{i=1}^{K}$. Specifically, while the matrix $Q_i = C^{-1} \Delta_{ii} C'^{-1}$ is independent of $\sigma$, the asymptotic proportions $\lambda_i \propto \sigma_i \cdot [Q_i]_{11}$ do vary with $\sigma$. In this proof, we write $\lambda_i(\sigma)$ to highlight this dependence. 

Next, we recall the matrix $M_j$ introduced earlier in (\ref{eq:Mj}). We note that for fixed matrices $V^0$ and $C$, each entry of $M_j(\sigma)$ is a fixed linear combination of $\frac{\sigma_1^2}{\lambda_1(\sigma)}, \dots, \frac{\sigma_K^2}{\lambda_K(\sigma)}$. 

Then, the parameter $a_j(\sigma)$ in (\ref{eq:firstDerApprox}) is given by (see the previous proof)
\begin{equation}\label{eq:aj}
a_j(\sigma) = - \frac{1}{2} -\frac{\lambda_j(\sigma) \cdot [M_j(\sigma)]_{11}}{2 [Q_j]_{11}} 
= - \frac{1}{2} +  \lambda_j(\sigma) \sum_{i=1}^{K} \tilde{b}_{ji}\frac{\sigma_i^2}{\lambda_i(\sigma)}
= - \frac{1}{2} + \sum_{i = 1}^{K} b_{ji} \sigma_i \sigma_j
\end{equation}
for some constants $\tilde{b}_{ji}, b_{ji}$ independent of $\sigma$. In the last equality above, we used the fact that $\frac{\lambda_j(\sigma)}{\lambda_i(\sigma)}$ equals a constant times $\frac{\sigma_j}{\sigma_i}$. 

Thus Lemma \ref{lemmFirstDerApprox} gives
\[
f(q_j, q_{-j}) - f(q_j+1, q_{-j}) = \frac{\sigma_j^2 \cdot [Q_j]_{11}}{(q_j-a_j(\sigma))^2} + O\left( \frac{1}{t^4}; c_0 \right)
\]
whenever $|q_i - \lambda_i(\sigma) \cdot t | \leq c_0, \forall i$. We comment that the Big O constant here may depend on $\sigma$. However, a single constant suffices if we restrict each $\sigma_i$ to be bounded above and bounded away from zero. Since measure-zero sets are closed under countable unions, this restriction does not affect the result we want to prove. 

By the above approximation, a necessary condition for $t \in \mathcal{A}^*(c_0, c_1)$ is that $q_1, q_2, q_3$ satisfy 
\begin{equation}\label{eq:Dio2}
\left|\left(q_2 - a_2(\sigma) \right)  - \frac{\eta \cdot \sigma_{2}}{\sigma_{1}} \left(q_1 - a_1(\sigma) \right)\right| \leq \frac{c_6}{q_1}
\end{equation}
as well as 
\begin{equation}\label{eq:Dio3}
\left|\left(q_3 - a_3(\sigma) \right)  - \frac{\kappa \cdot \sigma_{3}}{\sigma_{1}} \left(q_1 - a_1(\sigma) \right)\right| \leq \frac{c_6}{q_1}
\end{equation}
for some constant $c_6$ independent of $\sigma$ ($c_6$ may depend on $c_0, c_1$ stated in the lemma). The constant $\eta$ is given by $\eta = \sqrt{[Q_2]_{11} / [Q_1]_{11}}$, and similarly for $\kappa$. 

It remains to show that for generic $\sigma$, there are only finitely many positive integer triples $(q_1, q_2, q_3)$ satisfying the \emph{simultaneous Diophantine approximation} (\ref{eq:Dio2}) and (\ref{eq:Dio3}). To prove this, we assume that each $\sigma_i$ is i.i.d. drawn from the uniform distribution on $[\frac{1}{L}, L]$, where $L$ is a large constant. Denote by $F(q_1, q_2, q_3)$ the event that (\ref{eq:Dio2}) and (\ref{eq:Dio3}) hold simultaneously. We claim that there exists a constant $c_7$ such that $\mathbb{P}(F(q_1, q_2, q_3)) \leq \frac{c_7}{q_1^4}$ holds for all $q_1, q_2, q_3$. Since $F(q_1, q_2, q_3)$ cannot occur for $q_2, q_3 > c_8 q_1$, this claim will imply 
\begin{equation}\label{eq:BorelCantelli}
\sum_{q_1, q_2, q_3} \mathbb{P}(F(q_1,q_2,q_3)) < \sum_{q_1} \sum_{q_2, q_3 \leq c_8q_1} \frac{c_7}{q_1^4} < \sum_{q_1} \frac{c_7c_8^2}{q_1^2} < \infty.
\end{equation}
Generic finiteness of tuples $(q_1, q_2, q_3)$ will then follow from the Borel-Cantelli Lemma.\footnote{Because of the use of Borel-Cantelli Lemma, this proof (unlike Lemma \ref{lemmTwoDivisions} above) does not allow us to effectively determine for given $\sigma$ whether (\ref{eq:Dio2}) and (\ref{eq:Dio3}) only have finitely many integer solutions. Nonetheless, a modification of this proof does imply the following finite-time probabilistic statement: when $\sigma_1, \dots, \sigma_K$ are independently drawn, the probability that the optimality strategy coincides with $t$-optimality at every period $t \geq T$ is at least $1 - O(\frac{1}{T})$, where the constant involved only depends on the distribution of $\sigma$.}

To prove this claim, it suffices to show that if $\sigma = (\sigma_1, \sigma_2, \sigma_3, \sigma_4, \dots, \sigma_K)$ and $\sigma' = (\sigma_1, \sigma_2', \sigma_3', \sigma_4, \dots, \sigma_K)$ both satisfy (\ref{eq:Dio2}) and (\ref{eq:Dio3}), then $|\sigma_2 - \sigma_2'|, |\sigma_3 - \sigma_3'| \leq \frac{c}{q_1^2}$ for some constant $c$.\footnote{This implies that the probability of the event $F(q_1, q_2, q_3)$ conditional on any value of $\sigma_1, \sigma_4, \dots, \sigma_K$ is bounded by $\frac{c_7}{q_1^4}$, which is stronger than the claim.}
Without loss, we assume $|\sigma_2 - \sigma_2'| \geq |\sigma_3 - \sigma_3'|$. Using (\ref{eq:aj}), we can rewrite the condition (\ref{eq:Dio2}) as 
\[
\bigg| \underbrace{\left(q_2 + \frac{1}{2}\right) - \frac{\eta \cdot \sigma_{2}}{\sigma_{1}}\left(q_1 + \frac{1}{2}\right) }_{A}+ \underbrace{\sum_{i}\beta_i \sigma_2 \sigma_i}_{B} \bigg| \leq \frac{c_6}{q_1}
\]
for some constants $\beta_i$ independent of $\sigma$. A similar inequality holds at $\sigma'$: 
\[
\bigg| \underbrace{\left(q_2 + \frac{1}{2}\right) - \frac{\eta \cdot \sigma_{2}'}{\sigma_{1}}\left(q_1 + \frac{1}{2}\right)}_{A'} + \underbrace{\sum_{i}\beta_i \sigma_2' \sigma_i'}_{B'} \bigg| \leq \frac{c_6}{q_1}.
\]
It follows from the above two inequalities that $|A+B - A'-B'| \leq \frac{2c_6}{q_1}$. Furthermore, since $|A-A'| \leq |A+B - A'-B'| + |B-B'|$ (by triangle inequality), we deduce
\begin{equation}\label{eq:sigmaDifference}
\left|\frac{\eta \cdot (\sigma_{2}' - \sigma_2)}{\sigma_{1}}\cdot \left(q_1 + \frac{1}{2}\right) \right| \leq \frac{2c_6}{q_1} + \left| \sum_{i} \beta_i(\sigma_2'\sigma_i' - \sigma_2 \sigma_i) \right|.
\end{equation}
Because $\sigma_i' = \sigma_i$ for $i \neq 2, 3$, we have 
\begin{align*}
\left| \sum_{i} \beta_i(\sigma_2'\sigma_i' - \sigma_2 \sigma_i) \right| 
&= \left|  \sum_{i} \beta_i (\sigma_2' - \sigma_2) \sigma_i + \sum_{i} \beta_i \sigma_2'(\sigma_i'-\sigma_i) \right| \\
&= \left| \left(\sum_{i} \beta_i (\sigma_2' - \sigma_2) \sigma_i\right) + \beta_2 \sigma_2' (\sigma_2' - \sigma_2) + \beta_3 \sigma_2'(\sigma_3' - \sigma_3) \right| \\
& \leq (K+2)L \cdot \max_{i} \left| \beta_i \right| \cdot \left| \sigma_2' - \sigma_2 \right|.
\end{align*}
Plugging this estimate into (\ref{eq:sigmaDifference}), we obtain the desired result $|\sigma_2 - \sigma_2'| \leq \frac{c}{q_1^2}$. This completes the proof of the lemma. 
\end{proof}

\subsubsection{Monotonicity of $t$-Optimal Divisions}
We apply Lemma \ref{lemmThreeDivisions} to prove the eventual monotonicity of $t$-optimal divisions in generic informational environments. 
\begin{proposition}\label{propMonotone}
Fix $V^0$ and $C$. For generic signal variances $\{\sigma_i^2\}_{i=1}^{K}$, there exists $T_0$ such that for $t \geq T_0$, the $t$-optimal division $n(t)$ is unique, and it satisfies $n_i(t+1) \geq n_i(t), \forall i$. 
\end{proposition}

\begin{proof}
Uniqueness follows from the stronger fact that in generic informational environments, $f(q_1, \dots, q_K)$ differs from $f(q_1', \dots, q_K')$ whenever $q \neq q'$. Below we focus on monotonicity.

Using the order difference lemma, we can already deduce the difference $|n_i(t+1) - n_i(t)|$ is no more than $1$ 
at sufficiently late periods $t$. Suppose that $n_1(t+1) = n_1(t) - 1$. Then because $\sum_{i} (n_i(t+1) - n_i(t)) = 1$, we can without loss assume $n_2(t+1) = n_2(t) + 1$ and $n_3(t+1) = n_3(t) + 1$. 

For notational ease, write $n_i = n_i(t), n_i' = n_i(t+1)$. By $t$-optimality, we have
\[
f(n_1, n_2, n_3, \dots, n_K) \leq f(n_1-1, n_2+1, n_3, \dots, n_K) 
\]
\[
f(n_1', n_2', n_3', \dots, n_K') \leq f(n_1'+1, n_2'-1, n_3', \dots, n_K')
\]
These inequalities are equivalent to 
\begin{equation}\label{eq:ineq5}
\partial_2 f(n_1-1, n_2, n_3, \dots, n_K) \geq \partial_1 f(n_1-1, n_2, n_3, \dots, n_K)
\end{equation}
\begin{equation}\label{eq:ineq6}
\partial_2 f(n_1', n_2'-1, n_3', \dots, n_K') \leq \partial_1 f(n_1', n_2'-1, n_3', \dots, n_K')
\end{equation}
with $\partial_i f$ representing the \emph{discrete partial derivative}. 

Since $n_2' - 1 = n_2$, the LHS of (\ref{eq:ineq6}) is at least the LHS of (\ref{eq:ineq5}) minus a number of cross partials. Similarly, the RHS of (\ref{eq:ineq6}) is at most bigger than the RHS of (\ref{eq:ineq5}) by a number of cross partials. Thus the only way (\ref{eq:ineq5}) and (\ref{eq:ineq6}) can both hold is if the two sides of (\ref{eq:ineq5}) differ by no more than $O(\frac{1}{t^4})$. That is, for some absolute constant $c_1$,\footnote{As discussed in the proof of Lemma \ref{lemmThreeDivisions}, we can find a single constant $c_1$ that works for all $\sigma$ bounded above and bounded away from zero.} we have
\begin{equation}\label{eq:tinyDiff2}
|f(n_1 - 1, n_2+1, n_3, \dots, n_K) - f(n_1, n_2, n_3, \dots, n_K)| \leq \frac{c_1}{t^4}.
\end{equation}
An analogous argument yields
\begin{equation}\label{eq:tinyDiff3}
|f(n_1 - 1, n_2, n_3+1, \dots, n_K) - f(n_1, n_2, n_3, \dots, n_K)| \leq \frac{c_1}{t^4}.
\end{equation}
\noindent But now we can apply Lemma \ref{lemmThreeDivisions} to show that in generic environments, there are only finitely many integer tuples $(n_1, \dots, n_K)$ that satisfy both (\ref{eq:tinyDiff2}) and (\ref{eq:tinyDiff3}). This proves the result.
\end{proof}

\subsubsection{Completing the Proof of Theorem \ref{propGenericEventual}}
By Proposition \ref{propMonotone}, generically there exists $T_0$ such that $n(t)$ is monotonic in $t$ after $T_0$ periods. Thus, using our dynamic Blackwell lemma, if the DM achieves $t$-optimality at some period $t \geq T_0$, he will continue to do so in the future. By Proposition \ref{propGenericAlmostAll}, such a time $t$ does exist. This proves Theorem \ref{propGenericEventual}.

\subsection{Proof of Proposition \ref{propBoundOnB} (Bound on $B$)}

\subsubsection{Preliminary Estimates}
Throughout, we work with the linearly-transformed model, where each signal $X_i$ is simply $\tilde{\theta}_i$ plus standard Gaussian noise, and the DM's prior covariance matrix over the transformed states is $\tilde{V}$. Let $\gamma = \gamma(q_1, \dots, q_K)$ represent the following $K \times 1$ vector:
\begin{equation}\label{eq:gamma}
\gamma = (\tilde{V}+E)^{-1}\cdot \tilde{V} \cdot w
\end{equation}
with $E = \diag(\frac{1}{q_1}, \dots, \frac{1}{q_K})$. For $1 \leq i \leq K$, $\gamma_i$ denotes the $i$-th coordinate of $\gamma$. 

Here we re-derive the posterior variance function $f$, its (usual continuous) derivatives and second derivatives. Our formulae below take as primitives $\tilde{V}$ and $w$, but they are equivalent to those presented in Appendix \ref{appxPrelim} (for the original model). 

\begin{fact}[Posterior Variance] $f(q_1, \dots, q_K) = w'(\tilde{V} - \tilde{V}(\tilde{V}+E)^{-1}\tilde{V})w.$
\end{fact}

\begin{fact}[Partial Derivatives of Posterior Variance]
$
\partial_i f(q_1, \dots, q_K) = - \frac{1}{q_i^2} \cdot w' \tilde{V}(\tilde{V}+E)^{-1} \Delta_{ii} (\tilde{V}+E)^{-1} \tilde{V} w = - \frac{\gamma_i^2}{q_i^2}.
$
\end{fact}

\begin{fact}[Second-Order Partial Derivatives of Posterior Variance]
\begin{align*}
&\partial_{ii} f(q_1, \dots, q_K)~\nonumber\\
&~~~~=~ \frac{2\cdot w' \tilde{V}(\tilde{V}+E)^{-1} \Delta_{ii} (\tilde{V}+E)^{-1} \tilde{V} w}{q_i^3}  - \frac{2\cdot w' \tilde{V}(\tilde{V}+E)^{-1} \Delta_{ii} (\tilde{V}+E)^{-1} \Delta_{ii} (\tilde{V}+E)^{-1} \tilde{V} w}{q_i^4} \nonumber\\
&~~~~=~ \frac{2 \gamma_i^2}{q_i^3} \cdot \left(1 - \frac{[(\tilde{V}+E)^{-1}]_{ii}}{q_i}\right)
\end{align*}
\end{fact}

\begin{fact}[Cross-Partial Derivatives of Posterior Variance]
\begin{align*}
\partial_{ij} f(q_1, \dots, q_K)
=~& \frac{-2}{q_i^2 q_j^2} \cdot w' \tilde{V}(\tilde{V}+E)^{-1} \Delta_{ii} (\tilde{V}+E)^{-1} \Delta_{jj} (\tilde{V}+E)^{-1} \tilde{V} w \nonumber\\
=~& \frac{-2 \gamma_i \gamma_j}{q_i^2 q_j^2} \cdot [(\tilde{V}+E)^{-1}]_{ij}.
\end{align*}
\end{fact}

All of the above facts can be proved by simple linear algebra, so we omit the details.

\subsubsection{Refined Asymptotic Characterization of $n(t)$}
We now specialize to $w = \mathbf{1}$ and establish the next lemma, which refines our asymptotic characterization of $n(t)$ in Appendix \ref{appxBatch}.\footnote{Easy to see that in the transformed model, $\lambda_i \propto \lvert w_i \rvert$. So $\lambda_i = \frac{1}{K}$ here.} Proposition \ref{propBoundOnB} will immediately follow. 

\begin{lemma}\label{lemmFreqSharp}
For $t \geq 8(R+1)K\sqrt{K}$, it holds that $\lvert n_i(t) - \frac{t}{K} \rvert \leq 4(R+1)\sqrt{K}$. 
\end{lemma}

\begin{proof}
Note from (\ref{eq:gamma}) that $(\tilde{V}+E)\gamma = \tilde{V}w$. So $\tilde{V}(w - \gamma) = E \gamma = (\frac{\gamma_1}{q_1}, \dots, \frac{\gamma_K}{q_K})'$, and 
\[
w - \gamma = (\tilde{V})^{-1} \cdot \left(\frac{\gamma_1}{q_1}, \dots, \frac{\gamma_K}{q_K}\right)'. 
\]
From the definition of the operator norm, we deduce
\begin{equation}\label{eq:gammaIneq1}
\sum_{i=1}^{K} (1 - \gamma_i)^2 = \|w - \gamma\|^2 \leq R^2 \cdot \left( \sum_{j=1}^{K} \frac{\gamma_j^2}{q_j^2} \right).
\end{equation}
This holds for any division vector $q$ and the corresponding $\gamma$ (which is a function of $q$). 

Now suppose without loss of generality that $n_1(t) \geq \frac{t}{K}$. Let $q = (n_1(t)-1, n_2(t), \dots, n_K(t))$ and consider the corresponding $\gamma$. Then from $t$-optimality we have
\[
\lvert f(q_1 + 1, q_{-1}) - f(q) \rvert \geq \lvert f(q_j+1, q_{-j}) - f(q) \rvert, \quad \forall j. 
\]
Note that the discrete partial derivatives above are related to the usual continuous partials by the following inequalities:\footnote{The RHS follows from the convexity of $f$; the LHS can be proved by using Fact 2, Fact 3 and noting that $\gamma_j^2$ is an increasing function in $q_j$, because $\frac{\partial \gamma_j(q)}{\partial q_j} = \frac{\gamma_j}{q_j^2} \cdot [(V+E)^{-1}]_{jj}$ has the same sign of $\gamma_j$.}
\[
\frac{\gamma_j^2}{q_j(q_j+1)} \leq \lvert f(q_j+1, q_{-j}) - f(q) \rvert \leq \frac{\gamma_j^2}{q_j^2}.
\]
We therefore deduce 
\begin{equation}\label{eq:gammaIneq2}
\frac{\gamma_1^2}{q_1^2} \geq \frac{\gamma_j^2}{q_j(q_j+1)}, \quad \forall j. 
\end{equation}
Combining (\ref{eq:gammaIneq1}) and (\ref{eq:gammaIneq2}) and using $\frac{1}{q_j^2} \leq \frac{2}{q_j(q_j+1)}$, we see that 
\begin{equation}\label{eq:gammaIneq3}
\sum_{i = 1}^{K} (1- \gamma_i)^2 \leq 2R^2K \cdot \frac{\gamma_1^2}{q_1^2}.
\end{equation}

In particular, we know that $\gamma_1 - 1 \leq R \sqrt{2K} \cdot \frac{\gamma_1}{q_1}$. Easy to see this implies 
\begin{equation}\label{eq:gammaIneq4}
\gamma_1 \leq 1 + \frac{2R\sqrt{K}}{q_1} \leq \sqrt{2}
\end{equation}
whenever $q_1 = n_1(t)-1 \geq \frac{t}{K} - 1 \geq (2\sqrt{2}+2)R\sqrt{K}$. Plugging this back into the RHS of (\ref{eq:gammaIneq3}), we then obtain
\begin{equation}\label{eq:gammaIneq5}
\gamma_j \geq 1 - \frac{2R\sqrt{K}}{q_1} \geq2-\sqrt{2}.
\end{equation}

Now use (\ref{eq:gammaIneq2}), (\ref{eq:gammaIneq4}) and (\ref{eq:gammaIneq5}) to deduce that 
\[
q_j+1 \geq \frac{\gamma_j}{\gamma_1} \cdot q_1 \geq \frac{1 - \frac{2R\sqrt{K}}{q_1}}{1 + \frac{2R\sqrt{K}}{q_1}} \cdot q_1 \geq \left(1 - \frac{4R\sqrt{K}}{q_1}\right) \cdot q_1 = q_1 - 4R\sqrt{K}.
\]
Recall $q_j = n_j(t)$ for $j > 1$ and $q_1 = n_1(t) - 1$. We thus have 
\begin{equation}\label{eq:gammaIneq6}
n_j(t) \geq n_1(t) - 4R \sqrt{K} - 2.
\end{equation}
Since $n_1(t) \geq \frac{t}{K}$, the above implies $n_j(t) \geq \frac{t}{K} - 4(R+1)\sqrt{K}$ for each signal $j$. This proves half of the lemma. 

For the other half, note that $n_j(t) \leq \frac{t}{K}$ must hold for \emph{some} signal $j$. Thus (\ref{eq:gammaIneq6}) yields $n_1(t) \leq \frac{t}{K} + 4(R+1)\sqrt{K}$. This is not just true for signal $1$, but in fact for any signal $i$ with $n_i(t) \geq \frac{t}{K}$. So we conclude $n_i(t) \leq \frac{t}{K} + 4(R+1)\sqrt{K}$ for each signal $i$. The proof of the lemma is complete. 
\end{proof}


%% file: switch-fig.tex
\ifx\du\undefined
  \newlength{\du}
\fi
\setlength{\du}{15\unitlength}
\begin{tikzpicture}
\pgftransformxscale{1.000000}
\pgftransformyscale{-1.000000}
\definecolor{dialinecolor}{rgb}{0.000000, 0.000000, 0.000000}
\pgfsetstrokecolor{dialinecolor}
\definecolor{dialinecolor}{rgb}{1.000000, 1.000000, 1.000000}
\pgfsetfillcolor{dialinecolor}
\pgfsetlinewidth{0.100000\du}
\pgfsetdash{}{0pt}
\pgfsetdash{}{0pt}
\pgfsetbuttcap
{
\definecolor{dialinecolor}{rgb}{0.000000, 0.000000, 0.000000}
\pgfsetfillcolor{dialinecolor}
\pgfsetarrowsstart{latex}
\pgfsetarrowsend{latex}
\definecolor{dialinecolor}{rgb}{0.000000, 0.000000, 0.000000}
\pgfsetstrokecolor{dialinecolor}
\pgfpathmoveto{\pgfpoint{16.80062\du}{8.900047\du}}
\pgfpatharc{307}{234}{6.300000\du and 6.300000\du}
\pgfusepath{stroke}
}
\definecolor{dialinecolor}{rgb}{0.000000, 0.000000, 0.000000}
\pgfsetstrokecolor{dialinecolor}
\node[anchor=west] at (-2.70000\du,10.000000\du){$\underbrace{s_1~~s_2~~s_3~~s_4~~s_5~~\ldots~~s_{t_0-1}}_{\text{Signals match divisions $d^*(0),\ldots,d^*(t_0-1)$}}X_i~~\underbrace{s_{t_0+1}~~\ldots~~~s_{\tau-1}}_{\text{None of these are $X_j$}}~~X_j~~s_{\tau+1}~~\ldots$};
\definecolor{dialinecolor}{rgb}{0.000000, 0.000000, 0.000000}
\pgfsetstrokecolor{dialinecolor}
\node[anchor=west] at (11.00000\du,7.000000\du){\scriptsize{$(i,j)$-switch}};
\end{tikzpicture}

%% file: onlineappendix.tex
\section{Online Appendix}

\subsection{Applications of Results from Section \ref{games} (Multi-player Games)} \label{appxGames}

\subsubsection{Beauty Contest}
\cite{HellwigVeldkamp} introduced a beauty contest game with endogenous one-shot information acquisition. We build on this by modifying the information acquisition stage so that players \emph{sequentially} acquire information over many periods (rather than once), and face a \emph{capacity constraint} each period (rather than costly signals). We show that the basic insights of \cite{HellwigVeldkamp} hold in this setting.

Specifically, suppose that at an unknown final period, a unit mass of players simultaneously chooses prices $p_i \in \mathbb{R}$ to minimize the (normalized) squared distance between their price and an unknown target price $p^*$, which depends on the unknown state $\omega$ and also on the average price $\overline{p}=\int p_i ~ di$:
\begin{equation} \label{eq:HVpayoff}
u_i(p_i, \overline{p}, \omega) = -\frac{1}{(1-r)^2}\cdot (p_i - p^*)^2 \quad \mbox{ where } p^* = (1-r)\cdot \omega + r\cdot \overline{p}.
\end{equation}
The constant $r \in (-1,1)$ determines whether pricing decisions are complements or substitutes.\footnote{When $r > 0$, best responses are increasing in the prices set by other players, thus decisions are complements. Conversely, $r < 0$ implies decisions are substitutes.}

In every period up until the final period, each player acquires $B$ signals from the set $(X^{i}_k)$, as in the framework we have developed. To closely mirror the setup in \cite{HellwigVeldkamp}, we set each $\theta^i_1=\omega$. Assuming ``conditional independence" of players' signals, we can directly apply Corollary \ref{corDominant} and conclude that in every equilibrium, players choose a deterministic (myopic) sequence of information acquisitions. This result echoes \cite{HellwigVeldkamp}, who show that equilibrium is unique when players choose from \emph{private} signals (see their Subsection 1.3.4).\footnote{\cite{HellwigVeldkamp} also study a case in which players observe signals that are distorted by a common noise (which violates conditional independence). They show that multiple equilibria generally arise with such ``public signals". \cite{DewanMyatt}, \cite{MyattWallace} and \cite{Pavan} restore a unique linear symmetric equilibrium by assuming perfectly divisible signals, similar to the continuous-time variant of our model. In contrast, our equilibrium analysis relies on the informational environment (i.e.\ conditional independence), but not on symmetry or linearity of the strategy.} Our extension is to introduce dynamics and show how the dynamic problem can be reduced into a static one.

Let $\Sigma(t)$ be the posterior variance about $\omega$ after the first $t$ myopic observations. Since the players in our model acquire $B$ signals each period, their (common) posterior variance at the end of $t$ periods is given by $\Sigma(Bt)$. Thus, conditional on period $t$ being the final period, our game is as if the players acquire a batch of $Bt$ signals and then choose prices. This means that equilibrium prices are determined in the same way as in \cite{HellwigVeldkamp}:
\begin{equation}\label{eq:HVprice}
p(\mathcal{I}^i_{\leq Bt}) = \frac{1-r} {1-r + r \cdot \Sigma(Bt)} \cdot \mathbb{E}(\omega \vert \mathcal{I}^i_{\leq Bt}),
\end{equation}
where $\mathcal{I}^i_{\leq Bt}$ represents player $i$'s information set, consisting of $Bt$ signal realizations.

We can use this characterization of equilibrium to re-evaluate the main insight in \cite{HellwigVeldkamp}: the incentive to acquire more informative signals is increasing in aggregate information acquisition if decisions are complements and decreasing if decisions are substitutes. For this purpose, we augment the model with a period $0$, in which each player $i$ invests in a capacity level $B_i$ at some cost. Afterwards, players acquire information myopically (under possibly differential capacity constraints) and participate in the beauty contest game. 

Let $\mu \in \Delta(\mathbb{Z_{+}})$ be the distribution over capacity levels chosen by player $i$'s opponents. Then, player $i$'s expected utility from choosing capacity $B_i$ is given by 
\begin{equation*}\label{eq:HVeu}
EU(B_i, \mu) = -\mathbb{E}_{t \sim \pi} \left[\frac{\Sigma(B_it)}{\left(1-r + r\cdot \int_{B} \Sigma(Bt) ~d\mu(B)\right)^2}\right].
\end{equation*}
Above, the expectation is taken with respect to the random final period $t$ distributed according to $\pi$, while inside the expectation, the term $\int_{B} \Sigma(Bt) ~d\mu(B)$ is the average posterior variance among the players. Similar to Proposition 1 in \cite{HellwigVeldkamp}, we have the following result:

\begin{corollary}\label{corHV}
Suppose $\hat{B_i} > B_i$ and $\hat{\mu} > \mu$ in the sense of first-order stochastic dominance. Then the sign of the difference $EU(B_i, \mu) + EU(\hat{B_i}, \hat{\mu}) - EU(B_i, \hat{\mu}) - EU(\hat{B_i}, \mu)$ is
\begin{enumerate}
\itemsep0em 
\item[(a)] zero, if there is no strategic interaction $(r=0$); 
\item[(b)] positive, if decisions are complementary ($r>0)$;
\item[(c)] negative, if decisions are substitutes ($r<0)$. 
\end{enumerate}
\end{corollary}

\noindent When decisions are complements, the value of additional information is increasing in the amount of \emph{aggregate} information. Thus player $i$ has a stronger incentive to choose a higher signal capacity if his opponents (on average) acquire more signals. This incentive goes in the opposite direction when decisions are substitutes, which confirms the main finding of \cite{HellwigVeldkamp}.

\subsubsection{Strategic Trading}
We consider the strategic trading game introduced in \cite{LambertOstrovskyPanov}, in which individuals trade given asymmetric information about the value of an asset. We endogenize the information available to traders by adding a pre-trading stage in which traders sequentially acquire signals. As before, we suppose that trading occurs at a final time period that is determined according to an arbitrary full-support distribution.

In more detail: At the final time period, a security with unknown value $v$ is traded in a market, and each of $n$ traders submits a demand $d_i$. There are additionally liquidity traders who generate exogenous random demand $u$. A market-maker privately observes a signal $\theta_M$ (possibly multi-dimensional) and the total demand $D = \sum_i d_i + u$. He sets the price $P(\theta_M, D)$, which in equilibrium equals $\mathbb{E}[v \mid \theta_M, D]$. Each strategic trader then obtains profit $\Pi_i = d_i \cdot (v- P(\theta_M,D))$.

We suppose that in each period up to and including the final time period, each trader $i$ chooses to observe a signal from his set $(X^{i}_k)$ (described in Section \ref{games}). The requirement of conditional independence is strengthened to apply to a payoff-relevant \emph{vector} $\omega=(v,\theta_M,u)$ (instead of a real-valued unknown): That is, for each player $i$, conditional on the value of $\theta^i_1$, the payoff-relevant vector $\omega$ and the other players' unknown states $(\theta^j)_{j \neq i}$ are assumed to be conditionally independent from player $i$'s states $\theta^i$. Relative to the fully general setting considered in \cite{LambertOstrovskyPanov}, this assumption allows for flexible correlation \emph{within} a player's signals, but places a strong restriction on the correlation \emph{across} different players' signals. Applying Corollary \ref{corDominant}, we can conclude that:

\begin{corollary}
Under the above assumptions, there is an essentially unique linear NE in which the on-path signal acquisitions are myopic, and in the final period, players play the unique linear equilibrium described in \cite{LambertOstrovskyPanov}.
\end{corollary}

Thus, the closed-form solutions that are a key contribution of \cite{LambertOstrovskyPanov} extend to our dynamic setting with endogenous information.

\subsection{Example in Which $n(t)$ is Not Monotone Even for Large $t$} \label{appxExample2}
Here we continue to study Example 2 presented in Figure \ref{fig:intuition} of the main text. We will show that the $t$-optimal division vectors $n(t)$ fail to be monotone, even when we consider only periods after some large $T$.

The posterior variance function is 
\[
f(q_1, q_2, q_3) = 1 - \frac{1}{1 + \frac{1}{q_1} + 1 - \frac{1}{1 + \frac{1}{q_2} +\frac{1}{1+q_3}}}
\]
This suggests that the $t$-optimal problem can be separated into two parts: choosing $q_1$, and allocating the remaining observations between $q_2$ and $q_3$. The latter allocation problem is simple: an optimal division satisfies $q_3 = q_2 - 1$ or $q_3 = q_2$. With some extra algebra, we obtain that for $N \geq 1$: 
\begin{enumerate}
\item If $t = 3N+1$, then the unique $t$-optimal division is $(N+2, N, N-1)$; 

\item If $t = 3N+2$, then the unique $t$-optimal division is $(N+3, N, N-1)$;

\item If $t = 3N+3$, then the unique $t$-optimal division is $(N+2, N+1, N)$. 
\end{enumerate}

Crucially, note that when transitioning from $t = 3N+2$ to $t = 3N+3$, the $t$-optimal number of $X_1$ signals is decreased. This reflects the complementarity between signals $X_2$ and $X_3$, which causes the DM to observe them in pairs. Due to this failure of monotonicity, a sequential rule cannot achieve the $t$-optimal division vectors for all large $t$.

\subsection{Additional Result for $K = 2$}\label{appxK=2} 
When there are only two states and two signals, we can show that for a broad class of environments, the myopic information acquisition strategy is optimal from period $1$.\footnote{In the proposition below, if the linear coefficients $a, b, c, d$ were picked at random, then with probability $\frac12$ we would have $abcd \leq 0$.}

\begin{proposition}\label{propK=2Complementary} 
Suppose $K = 2$, the prior is standard Gaussian $(V^0 = \mathbf{I}_2)$, and both signals have variance $1$.\footnote{We make these simplifying assumptions so that the condition for immediate equivalence is easy to state and interpret.} Write $C = \left(\begin{array}{cc} a & b \\
c & d
\end{array}\right)$ 
and assume without loss that $|ad| \geq |bc|$. Then the optimal information acquisition strategy is myopic whenever the following inequality holds: 
\begin{equation}\label{eq:K=2Complementary}
(1+2b^2)\cdot \lvert ad-bc \rvert \geq \lvert ad + bc \rvert.
\end{equation}
In particular, this is true whenever $abcd \leq 0$. 
\end{proposition}

\noindent To interpret, (\ref{eq:K=2Complementary}) requires that the determinant of the matrix $C$, $ad - bc$, is not too small (holding other terms constant). Equivalently, the two vectors (in $\mathbb{R}^2$) defining the signals should not be close to collinear. This rules out situations where the two signals provide such similar information in the initial periods that they substitute one another. 

\begin{proof}
Under the assumptions, the DM's posterior variance about $\theta_1$ is computed to be 
\[
f(q_1, q_2) = \frac{1+b^2q_1+d^2q_2}{1+(a^2+b^2)q_1+(c^2+d^2)q_2+(ad-bc)^2q_1q_2}.
\]
Given $q_i$ observations of each signal $i$ in the past, the myopic strategy chooses signal $1$ if and only if $f(q_1+1, q_2) < f(q_1, q_2+1)$, which reduces to
\begin{equation}\label{eq:K=2MyopicRule}
\begin{split}
&(ad-bc)^2b^2q_1^2 + (1+b^2)(ad-bc)^2q_1 - (a^2d^2-b^2c^2)q_1+ c^2(1+b^2)\\
<&(ad-bc)^2d^2q_2^2 + (1+d^2)(ad-bc)^2q_2 + (a^2d^2-b^2c^2)q_2+ a^2(1+d^2)
\end{split}
\end{equation}
The condition $|ad| \geq |bc|$ ensures that the RHS is an increasing function of $q_2$, because the coefficients in front of $q_2^2$ and $q_2$ are both positive. Meanwhile, the condition $(1+2b^2)|ad-bc| \geq |ad+bc|$ implies the LHS is larger when $q_1 = 1$ than when $q_1 = 0$, so that the LHS is also increasing in (integer values of) $q_1$. 

Even if $f$ may not be written into separable form, (\ref{eq:K=2MyopicRule}) suggests that the comparison between the marginal values of signal $1$ and $2$ ``is separable." It follows that the $t$-optimal division vectors are increasing in $t$. Proposition \ref{propK=2Complementary} is proved. 
\end{proof}

\subsection{Eventual Optimality of the Myopic Strategy}\label{appxMyopic}
Below, write $m(t)$ for the division vector at time $t$ achieved under the (history-independent) myopic rule.\footnote{That is, $m(t)=(m_1(t), \dots, m_K(t))$ where $m_i(t)$ is the number of times signal $i$ has been observed under myopic information acquisition prior to and including period $t$.} We have discussed that when Theorems \ref{propBatch} or \ref{propSeparable} apply, the myopic division vector $m(t)$ is $t$-optimal at every period $t$. In this appendix, we argue that generically, division vectors $m(t)$ at late periods are $t$-optimal. This result complements our Theorem \ref{propGenericEventual}, and suggests that a DM who naively follows the myopic rule all the way cannot do very poorly. 

To avoid repetition, here we only sketch the core argument. The main new step is to show that the division vectors $m(t)$ under the myopic rule grow to infinity in each coordinate; that is, a myopic DM would \emph{not get stuck} observing a subset of signals. Once this is shown, we can repeat the (rest of the) proof of Lemma \ref{lemmDynamicFreq} and deduce that $m_i(t) - \lambda_i \cdot t$ remains bounded. And with these asymptotic characterizations, we can reproduce the proof of Theorem \ref{propGenericEventual} (now for the myopic strategy instead of the optimal strategy) without trouble.\footnote{These latter steps are actually simpler to carry out for the myopic strategy. This is because in constructing a deviation from the myopic strategy, we only need to look for a lower posterior variance at a single period. So we no longer need to make use of switch deviations.}

To see myopic signal choices never get stuck, we establish the following lemma.
\begin{lemma}\label{lemmNotStuck}
Fix an arbitrary division vector $q \in \mathbb{R}_{+}^{K}$ (need not be integral). The partial derivatives of $f$ at $q$ are all zero if and only if $q_1 = \dots = q_K = \infty$. 
\end{lemma}

Intuitively, this holds because for normal linear signals, the posterior variance is \emph{globally convex}. So if each signal has zero marginal value relative to the division $q$, then $q$ must be a global minimizer of posterior variance. 

We conclude by mentioning that a similar result (i.e.\ myopic information acquisition does not get stuck) would not in general be true for other signal structures. The following is a counterexample with normal but non-linear signals. 
\begin{example}
Consider three states $\theta_1, \theta_2, \theta_3$ drawn independently. The DM has access to these three signals:
\begin{align*}
X_1 &= \theta_1 + \text{sign}(\theta_2) + \epsilon_1 \\
X_2 &= \text{sign}(\theta_2\theta_3) + \epsilon_2 \\
X_3 &= \theta_3 + \epsilon_3
\end{align*}
where $\epsilon_1, \epsilon_2, \epsilon_3$ are Gaussian noise terms. We focus on the prediction problem, in which (at a random time) the DM makes a prediction about $\theta_1$ and receives negative of the squared prediction error.

Note that prior to the first observation of $X_2$, signal $X_3$ is completely uninformative about the payoff-relevant state $\theta_1$ (even when combined with previous observations of $X_1$). Similarly, signal $X_2$ is individually uninformative about $\theta_2$,\footnote{This is because the sign of $\theta_2\theta_3$ does not contain any new information about $\theta_2$ when $\theta_3$ is equally likely to be positive or negative.} and thus about $\theta_1$. These imply that the DM's uncertainty about $\theta_1$ is not reduced upon the first observation of either $X_2$ or $X_3$. Hence, the myopic rule in this example is to always observe $X_1$, contrary to Lemma \ref{lemmNotStuck}. 

Thus, if the DM acquires information myopically, he will never completely learn the value of $\theta_1$. By contrast, if the DM is sufficiently patient, then his optimal strategy will observe each signal infinitely often and identify the value of $\theta_1$ in the long run. Thus, in this example the myopic signal path does not eventually agree with the optimal path. 
\end{example}